\documentclass{article}
\usepackage[utf8]{inputenc}
\usepackage{fullpage} 

\usepackage{amsmath}
\usepackage{amsfonts}
\usepackage{amssymb}
\usepackage{amsthm}

\usepackage{mathrsfs}

\usepackage[colorlinks]{hyperref}
  \hypersetup{linkcolor=blue,filecolor=blue,citecolor=blue,urlcolor=blue}
\usepackage{cleveref}
\usepackage[T1]{fontenc}
\usepackage{comment}
\usepackage{graphicx}
\usepackage[table,xcdraw]{xcolor}
\usepackage{tikz}

\usepackage{braket}

\usepackage[normalem]{ulem}

\newif\ifcomments
\commentsfalse

\newcommand{\alice}{\mathcal{A}}
\newcommand{\bob}{\mathcal{B}}
\newcommand{\charlie}{\mathcal{C}}
\newcommand{\abc}{\left(\mathcal{A}, \mathcal{B}, \mathcal{C}\right)}

\newcommand{\aliceprime}{\alice'}
\newcommand{\bobprime}{\bob'}
\newcommand{\charlieprime}{\charlie'}
\newcommand{\abcprime}{\brackets{\aliceprime, \bobprime, \charlieprime}}

\newcommand{\alicetild}{\widetilde{\alice}}
\newcommand{\bobtild}{\widetilde{\bob}}
\newcommand{\charlietild}{\widetilde{\charlie}}
\newcommand{\abctild}{\brac{\alicetild, \bobtild, \charlietild}}

\newtheorem{theorem}{Theorem}
\newtheorem{lemma}[theorem]{Lemma}
\newtheorem{claim}{Claim}

\newtheorem{corollary}{Corollary}
\newtheorem{remark}{Remark}
\newtheorem{definition}{Definition}
\theoremstyle{definition}

\newcommand{\bit}{\left\{0,1\right\}}
\newcommand{\uniform}{\xleftarrow{\$}}

\newcommand{\pr}[1]{\Pr \left[ #1 \right] }
\newcommand{\from}{\leftarrow}
\newcommand{\secparam}{\lambda}
\newcommand{\ct}{\mathsf{CT}}
\newcommand{\ketbra}[2]{\ket{#1}\!\!\bra{#2}}
\newcommand{\ketbraX}[1]{\ket{#1}\!\!\bra{#1}}

\newcommand{\negl}{\mathsf{negl}}
\newcommand{\poly}{\mathsf{poly}}
\newcommand{\tr}{\mathsf{Tr}}
\newcommand{\given}{\mid}
\DeclareMathOperator*{\E}{\mathbb{E}}

\newcommand{\id}{\mathsf{id}}
\newcommand{\ind}{\mathsf{ind}}

\newcommand{\abs}[1]{\left| #1 \right|}
\newcommand{\absbig}[1]{\bigg| #1 \bigg|}
\newcommand{\inner}[2]{\langle #1,#2\rangle}

\newcommand{\brackets}[1]{\left( #1 \right)}
\newcommand{\bracketsSquare}[1]{\left[ #1 \right]}
\newcommand{\bracketsCurly}[1]{\left\{ #1 \right\}}
\newcommand{\brac}{\brackets}
\newcommand{\bracC}{\bracketsCurly}
\newcommand{\bracS}{\bracketsSquare}

\newcommand{\enc}{\mathsf{Enc}}
\newcommand{\dec}{\mathsf{Dec}}

\newcommand{\gen}{\mathsf{Gen}}

\newcommand{\ue}{\mathsf{UE}}

\ifcomments
\newcommand{\edit}[1]{\textcolor{purple}{#1}}
\newcommand{\pnote}[1]{{\color{red} P:#1}}
\newcommand{\fatih}[1]{{\color{blue} F: #1}}
\newcommand{\hnote}[1]{{\color{magenta} HY: #1}}
\else
\newcommand{\edit}[1]{}
\newcommand{\pnote}[1]{}
\newcommand{\fatih}[1]{}
\newcommand{\hnote}[1]{}
\fi

\newcommand{\N}{\mathbb{N}}

\newcommand{\C}{\mathbb{C}}
\newcommand{\R}{\mathbb{R}}
\newcommand{\F}{\mathbb{F}}

\newcommand{\cH}{\mathcal{H}}

\newcommand{\cM}{\mathcal{M}}

\newcommand{\cR}{\mathcal{R}}
\newcommand{\cS}{\mathcal{S}}

\newcommand{\nldist}[2]{T_{SSI}\left(#1, \allowbreak #2\right)}

\newcommand{\tracedist}[2]{T\left(#1, #2\right)}
\newcommand{\haar}{\mathscr{H}}
\newcommand{\unitarygp}{\mathscr{U}}
\renewcommand{\braket}[2]{\langle #1 | #2 \rangle}

\newcommand{\qkey}{\ket{\includegraphics[height=1.5ex]{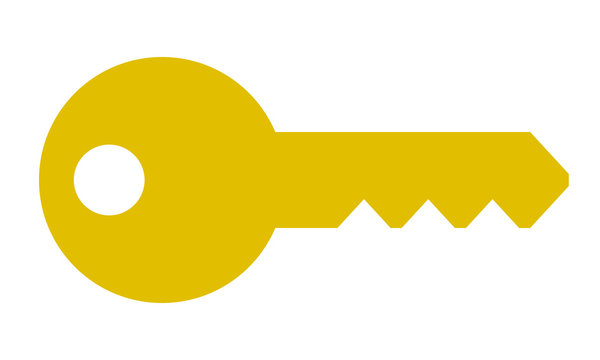}}\!}
\newcommand{\sqkey}{\ket{\includegraphics[height=1.2ex]{key.jpg}}\!}
\newcommand{\ekey}{{\sf ek}}

\newcommand{\distr}{{\cal D}}
\newcommand{\eps}{\varepsilon}
\newcommand{\adversary}{{\cal A}}
\newcommand{\bfB}{\textbf{B}}
\newcommand{\bfC}{\textbf{C}}
\newcommand{\prob}{{\sf Pr}}

\newcommand{\linop}{{\cal L}}

\newcommand{\qcipher}{\ket{\ct}}
\newcommand{\epr}{\mathsf{EPR}}
\newcommand{\cnot}{\mathsf{CNOT}}

\newcommand{\reg}[1]{{\color{gray}#1}}
\newcommand{\ignore}[1]{}

\renewcommand{\cal}[1]{\mathcal{#1}}

\newcommand{\ot}{\otimes}
\newcommand{\SymmSub}{\Pi_{\mathsf{Sym}}}

\newcommand{\Haar}{\mathrm{Haar}}
\newcommand{\Tr}{\mathsf{Tr}} 

\title{Simultaneous Haar Indistinguishability\\ with Applications to Unclonable Cryptography}
\author{Prabhanjan Ananth\thanks{prabhanjan@cs.ucsb.edu}\\ \small{UCSB} \and {Fatih Kaleoglu}\thanks{kaleoglu@ucsb.edu}\\ \small{UCSB} \and {Henry Yuen}\thanks{hyuen@cs.columbia.edu}\\ \small{Columbia University}}
\date{}

\begin{document}

\maketitle

\begin{abstract}
\noindent Unclonable cryptography is concerned with leveraging the no-cloning principle to build cryptographic primitives that are otherwise impossible to achieve classically. Understanding the feasibility of unclonable encryption, one of the key unclonable primitives, satisfying indistinguishability security in the plain model has been a major open question in the area. So far, the existing constructions of unclonable encryption are either in the quantum random oracle model or are based on new conjectures. 
\par We present a new approach to unclonable encryption via a reduction to a novel question about nonlocal quantum state discrimination: how well can non-communicating -- but entangled -- players distinguish between different distributions over quantum states? We call this task \emph{simultaneous state indistinguishability}. Our main technical result is showing that the players cannot distinguish between each player receiving independently-chosen Haar random states versus all players receiving the same Haar random state. 
\par We leverage this result to present the first construction of unclonable encryption satisfying indistinguishability security, with quantum decryption keys, in the plain model. We also show other implications to single-decryptor encryption and leakage-resilient secret sharing.

\end{abstract}
\newpage 
\tableofcontents
\newpage 

\newcommand{\ellpnorm}[2]{\left\|{#1}\right\|_{#2}}
\newcommand{\tvdist}[2]{d_{\mathsf{TV}}\brac{#1,#2}}
\newcommand{\bobx}{\widetilde{\text{Bob}}}
\newcommand{\charliex}{\widetilde{\text{Charlie}}}

\section{Introduction}
\label{sec:intro}

Quantum state discrimination~\cite{Hel69} is a foundational concept with applications to quantum information theory, learning theory and cryptography. In a state discrimination task, a party receives $\rho_x$ from an ensemble $\{\rho_x\}_{x \in {\cal X}}$ and has to determine which state it received. A compelling variant of this problem is concerned with the multi-party setting where there are two or more parties and each party receives a disjoint subset of qubits of $\rho_x$. This multi-party variant of state discrimination has also garnered interest from quantum information-theorists focused on the LOCC (local operations and classical communication) model and quantum data hiding~\cite{peres1991optimal,bennett1993teleporting,divincenzo2002quantum,chitambar2014everything}. Understanding the multi-party state discrimination problem in turn sheds light on the difficulty of simulating global measurements using local measurements \cite{RW12}.
\par An important aspect to consider when formulating the multi-party state discrimination problem is the resources shared between the different parties. If we allow the parties to share entanglement and also communicate with each other then this is equivalent to the original state discrimination task (against a single party) due to teleportation. Thus, in order for the multi-party setting to be distinct from the single-party setting, we need to disallow either shared entanglement or classical communication. This results in two different settings: 
\begin{itemize}
\item \textsc{Parties without shared entanglement}: In this setting, the parties are allowed to communicate using classical channels but they cannot share entanglement. The extensive research on quantum data hiding and LOCC~\cite{divincenzo2002quantum,GeaHiding02,hayden2005multiparty,MWW09,piani2014quantumness,bennett1999quantum,childs2013framework,chitambar2014asymptotic,chitambar2014everything,halder2019strong,nathanson2005distinguishing} are mainly concerned with this setting. 
\item \textsc{Parties with shared entanglement}: In this setting, the parties are allowed to share entanglement but they cannot communicate. On the contrary, this setting is relatively unexplored with the notable exception being the recent works of~\cite{MOST21,escola2022parallel}. There are good reasons to study multi-party state discrimination with shared entanglement. Firstly, it can be viewed as a subclass of semi-quantum games~\cite{buscemi2012all}, which are non-local games with quantum questions and classical answers. Secondly, it has connections to unclonable cryptography~\cite{Wiesner83,Aar09}, an emerging area in quantum cryptography, as discussed in~\cite{MOST21}. 
\end{itemize}

\paragraph{Our Work.} We focus on the setting when the parties share entanglement but are not allowed to communicate. We introduce a new concept called {\em simultaneous state indistinguishability} (SSI). In the two-party version of the problem, we have two parties (say, Bob and Charlie) who receive as input one of two bipartite states $\{\rho_0,\rho_1\}$ and they are supposed to distinguish. The first half of $\rho_b$ is given to Bob and the second half is given to Charlie for $b \in \bit$. 

We say that $\rho_0$ and $\rho_1$ are \emph{$\epsilon$-simultaneous state indistinguishable} if the probability that Bob and Charlie can simultaneously distinguish is at most $\epsilon$. That is, 
$\tvdist{x}{x'} \le \epsilon$, where $\tvdist{\cdot}{\cdot}$ denotes the total variation distance, 
$x = (x_B, x_C) \from (\text{Bob, Charlie})(\rho_0)$ denotes the random variable corresponding to the joint outputs of Bob and Charlie given input $\rho_0$, and similarly $x' = (x_B', x_C') \from (\text{Bob, Charlie})(\rho_1)$ denotes the random variable corresponding to their joint outputs when given input $\rho_1$.

This is related to a recent concept introduced by~\cite{MOST21} who considered the unpredictability (search) version of this problem whereas we are interested in indistinguishability. Looking ahead, for applications, it turns out that the indistinguishability notion is more amenable to carrying out proofs (e.g. in the security of unclonable encryption) compared to the unpredictability definition considered by~\cite{MOST21}. This is largely due to the fact that our notion is more compatible with the hybrid technique.
\par An interesting special case of simultaneous-state-indistinguishability is when Bob and Charlie either receive copies of the same state $\ket{\psi}$ drawn from some distribution $\distr$ or receive i.i.d. samples from $\distr$, i.e.
\begin{align*}
    \rho_0 = \E_{\ket{\psi} \sim \distr} {\psi}^{\otimes t} \ot \psi^{\ot t} \quad \text{vs.} \quad \rho_1 = \E_{\ket{\psi_B},\ket{\psi_C} \sim \distr} {\psi_B}^{\otimes t} \ot \psi_C^{\ot t}.
\end{align*}
Observe that if Bob and Charlie were allowed to make global entangled measurements then they can indeed distinguish by performing swap tests. However, it is not clear these two situations are distinguishable using local measurements, even with preshared entanglement between Bob and Charlie.

We study simultaneous state indistinguishability in the case that each party receives (copies of) a state drawn from the \emph{Haar measure}.\footnote{The Haar measure over states is the unique measure where for all fixed unitaries $U$, if $\ket{\psi}$ is Haar-distributed then so is $U \ket{\psi}$.} Specifically, we consider the following definition: 

\begin{quote}
    {{\bf $(d,t,\varepsilon)$-Simultaneous Haar Indistinguishability}}: We say that $(d,t,\varepsilon)$-simultaneous Haar indistinguishability holds if any two non-communicating and entangled adversaries Bob and Charlie can distinguish the following distributions with probability at most $\varepsilon$: 
\begin{itemize}
    \item Bob and Charlie each receive $t$ copies of $\ket{\psi}$, where $\ket{\psi}$ is a $d$-dimensional Haar state, 
    \item Bob receives $t$ copies of $\ket{\psi_B}$ and Charlie receives $t$ copies of $\ket{\psi_C}$, where $\ket{\psi_B},\ket{\psi_C}$ are i.i.d $d$-dimensional Haar states. 
\end{itemize} 
\end{quote}

\noindent In the default setting, both Bob and Charlie each output 1 bit. We also consider the setting when they output multiple bits. 

\par Variants of this problem have been studied in different contexts before. Independently, two works, namely, Harrow~\cite{Harrow23} and Chen, Cotler, Huang and Li~\cite{CCHL21} showed (for the case when $t=1$) that Bob and Charlie fail, except with probability negligible in the dimension $d$, in the above distinguishing experiment as long as they don't share any entanglement. In fact, Harrow's result proves something stronger: the indistinguishability holds even if the two parties exchange classical information (i.e., LOCC setting). Both works discuss the applications of this problem to well studied topics such as multiparty data hiding, local purity testing and separations between quantum and classical memory. Neither of the works~\cite{Harrow23,CCHL21} addresses the setting when Bob and Charlie can share an arbitrary amount of entanglement.

\paragraph{Main Result.} We show the following: 

\begin{theorem}[Informal]
For any $d,t \in \mathbb{N}$, $(d,t,\varepsilon)$-Simultaneous Haar indistinguishability holds for $\varepsilon = O\left( \frac{t^2}{\sqrt{d}} \right)$. 
\end{theorem}

\noindent Our result complements the works of~\cite{Harrow23,CCHL21} by showing that it is not possible to distinguish i.i.d versus identical Haar states either using the entanglement resource or using classical communication.
\par In the case when $t=1$, we show that $(d,1,\epsilon)$-simultaneous Haar indistinguishability does not hold for $\epsilon = O\left( \frac{1}{d} \right)$, which suggests that the above bound cannot be improved significantly. Perhaps surprisingly, our attack even holds in the setting when Bob and Charlie do not share any entanglement. This further indicates that for the problem of simultaneous Haar indistinguishability, the gap between the optimal success probabilities in the entangled and the unentangled cases is small. 

\paragraph{Applications.} Besides being a natural problem, simultaneous Haar indistinguishability has applications to unclonable cryptography~\cite{Wiesner83,Aar09,Zha19,BL20,AL20,CMP20,CLLZ21,CG23a}. This is an area of quantum cryptography that leverages the no-cloning principle of quantum mechanics to design cryptographic notions for tasks that are impossible to achieve classically. \\ 

\noindent \underline{\textsc{Unclonable Encryption}}: Unclonable encryption is an encryption scheme with quantum ciphertexts that are unclonable. It was first introduced by Broadbent and Lord~\cite{BL20} and is now considered a fundamental notion in unclonable cryptography. There are two security notions of unclonable encryption, namely search and indistinguishability security, studied in the literature. The search security stipulates that any cloning adversary\footnote{A cloning adversary is a tri-partite adversary $(\alice,\bob,\charlie)$. $\alice$ receives as input an unclonable state and produces a bipartite state given to $\bob$ and $\charlie$ who are not allowed to communicate. Then, in the challenge phase, the cloning adversary receives as input $({\sf ch}_{\bob},{\sf ch}_{\charlie})$ and then gives ${\sf ch}_{\bob}$ to $\bob$ and ${\sf ch}_{\charlie}$ to $\charlie$. Finally, $\bob$ and $\charlie$ output their respective answers. Refer to~\cite{AKL23} for an abstract modeling of unclonable security notions.} after receiving a ciphertext of randomly chosen message $m$ should not be able to guess $m$ except with probability negligible in $|m|$. In the challenge phase, the cloning adversary receives as input $(k,k)$, where $k$ is the decryption key. The indistinguishability security imposes a stronger guarantee that any cloning adversary after receiving encryption of $m_b$, for a randomly chosen bit $b$ and adversarially chosen message pair $(m_0,m_1)$, cannot predict $b$ except with probability negligibly close to $\frac{1}{2}$. 
\par While we have long known that search security is feasible~\cite{BL20}, establishing indistinguishability security has remained an important and intriguing open problem. Unclonable encryption satisfying indistinguishability is achievable in the quantum random oracle model~\cite{AKLLZ22,AKL23} or in the plain model (i.e., without oracles) based on new unproven conjectures~\cite{AB23}. Various generic transformations are also known that convert one-time unclonable encryption to the public-key variant~\cite{AK21} or those that convert unclonable encryption for one-bit messages to multi-bit messages~\cite{Hiroka23}. Given that unclonable-indistinguishable encryption has interesting applications to copy-protection~\cite{AK21,CG23}, it is important we completely settle its feasibility in the plain model. Perhaps embarassingly, we do not how to establish feasibility in the plain model even under strong assumptions such as indistinguishability obfuscation~\cite{BGIRSVY01,GGHRSW13}!
\par We consider the notion of unclonable encryption, where the decryption keys are allowed to be quantum keys. This only affects the challenge phase of the security experiment, where the cloning adversary now receives as input many copies\footnote{We consider a security notion where the adversary receives at most $t$ copies of the quantum decryption key and $t$ is fixed ahead of time.} of the quantum decryption key. 

We show the following: 

\begin{theorem}[Informal]
\label{thm:intro:ue}
There is an unclonable encryption scheme, with quantum decryption keys, satisfying indistinguishability security in the plain model.  
\end{theorem}

\noindent Ours is the first work to show that unclonable-indistinguishable encryption exists in the plain model albeit with quantum decryption keys. We note that this relaxation (i.e. decryption keys being quantum states) does not have a significant effect on the essence of unclonable encryption. A potential approach to transform this scheme into another scheme where the decryption keys are binary strings is to generate the decryption key to be an obfuscation of the setup algorithm that produces decryption keys. It is an intriguing problem to formalize the requirements of the underlying quantum obfuscation scheme. At the bare minimum, we require that the quantum obfuscation scheme satisfies the property that the obfuscated program can be represented as a binary string. While achieving quantum obfuscation has been a difficult open problem~\cite{BK21,BM22,BKNY23}, obfuscating special classes of quantum algorithms (those that capture the setup algorithm of~\Cref{thm:intro:ue}) could be relatively more tractable. 

\par Our scheme supports one-bit messages and is one-time secure. It is an interesting future direction to extend the works of~\cite{AK21} and~\cite{Hiroka23} to generically achieve  unclonable encryption with quantum decryption keys in the public-key setting and for longer messages. \\

\noindent \underline{\textsc{Single-Decryptor Encryption}}: Single-decryptor encryption~\cite{GZ20} is a sister notion of unclonable encryption, where instead of requiring the ciphertexts to be unclonable, we instead require the decryption keys to be unclonable. Constructions of single-decryptor encryption in different settings are known from a variety of assumptions~\cite{GZ20,CLLZ21,AKL23,KY23}. There are two important security notions considered in the literature. In the {\em independent} setting, in the challenge phase, the cloning adversary gets two independently generated ciphertexts while in the {\em identical} setting, it gets copies of the same ciphertext. All the known constructions of single-decryptor encryption~\cite{CLLZ21,AKL23,KY23} are in the independent setting and specifically, there are no known constructions in the identical setting. This should not be surprising in light of~\cite{GZ20} who showed the equivalence between single-decryptor encryption with identical ciphertexts and unclonable encryption, which suggests the difficulty in achieving the identical ciphertext setting.
\par We prove the following. 

\begin{theorem}[Informal]
There is a single-decryptor encryption scheme, with quantum ciphertexts, satisfying identical indistinguishability security.  
\end{theorem}

\noindent Ours is the first work to demonstrate the feasibility of single-decryptor encryption in the identical-challenge setting, albeit with quantum ciphertexts.
\par In addition to its applications to unclonable cryptography, simultaneous Haar indistinguishability can be used to construct leakage-resilient secret sharing. \\

\noindent \underline{\textsc{Leakage-Resilient Secret Sharing}}: Leakage-resilient cryptography~\cite{KR19} is an area of cryptography that is concerned with the goal of building cryptographic primitives that are resilient to side-channel attacks. We are interested in designing secret sharing schemes that are leakage resilient. In a leakage-resilient secret sharing scheme, a leakage function is applied on each share and we require the guarantee that all the leakages put together are not sufficient enough to compromise the security of the secret sharing scheme. 
Leakage-resilient secret sharing is a well studied topic~\cite{goyal2018non,srinivasan2019leakage,aggarwal2019stronger,CGGKLMZ20,benhamouda2021local} with applications to leakage-resilient secure multi-party computation~\cite{benhamouda2021local}. 
\par We consider a notion of leakage-resilient secret sharing, where we allow the parties holding the shares to be entangled with each other. We now require the guarantee that security should still hold even if each share is individually leaked. Moreover, we consider a relaxed requirement where the shares are allowed to be quantum. Just like the works in the classical setting, we consider the bounded leakage model. That is, if the number of qubits of each share is $m$ then we allow for some $\lfloor \frac{c}{n} \rfloor$ fraction of bits of leakage from each share, where $c$ is some constant and $n$ is the number of parties\footnote{We set $m \gg n$.}. 
\par We show the following: 

\begin{theorem}[Informal]
There is a 2-out-$n$ leakage-resilient secret sharing scheme with the following properties: (a) the shares are quantum, (b) the number of bits of leakage on each share is $\lfloor \frac{c \cdot m}{n} \rfloor$, where $c$ is some constant and the size of each share is $m$ qubits, and (c) the parties can share arbitrary amount of entanglement.  
\end{theorem}

\noindent In fact, our construction satisfies a stronger security guarantee where the adversary can receive $p(n)$ number of copies of its share, where $p(\cdot)$ is an arbitrary polynomial. 
\par A recent interesting work by~\cite{cakan2023unbounded} also considers leakage-resilient secret sharing schemes with quantum shares. However, there are notable differences. Firstly, they consider the setting when there can be unbounded amount of classical bits of leakage from each quantum share whereas we consider bounded leakage. On the other hand, we allow the parties to be entangled whereas they mainly focus on the LOCC setting. In fact, they show that it is not possible to achieve unbounded amount of leakage in the shared entanglement setting even with two parties; this is the reason we focus on the setting of bounded leakage. However, there seems to be a large gap between the amount of leakage leveraged in the impossibility result in~\cite{cakan2023unbounded}\footnote{Their impossibility result seems to require $O(2^m)$ bits of total leakage from all the parties.} and the leakage that we tolerate in our feasibility result. It is an interesting problem to close the gap. Finally, we allow each party to get arbitrary polynomially many copies of its share whereas~\cite{cakan2023unbounded} doesn't satisfy this guarantee. 

\subsection{Technical Overview}

\subsubsection{Simultaneous Haar Indistinguishability}
We formally define the notion of simultaneous state indistinguishability (SSI), of which simultaneous Haar indistinguishability is a special case. We consider a non-local distinguisher $(\text{Bob}, \text{Charlie}, \rho)$ where $\text{Bob}$ and $\text{Charlie}$ are spatially separated quantum parties who share an entangled state $\rho$. Given two distributions $\distr_1, \distr_2$ over bipartite states, we can write the distinguishing advantage of $(\text{Bob}, \text{Charlie}, \rho)$ as $\tvdist{x}{x'}$, where 
$x = (x_B,x_C)$ is the random variable corresponding to the output of Bob and Charlie when they get as input $\rho \otimes \psi$ where $\ket{\psi}$ is sampled from $\distr_1$. Here, $x_B$ refers to Bob's output and $x_C$ refers to Charlie's output. Similarly, $x' = (x_B',x_C')$ is the random variable corresponding to Bob and Charlie's outputs when they receive $\rho \otimes \psi'$ where $\ket{\psi'}$ is sampled from $\distr_2$. 

For fixed $\distr_1, \distr_2$ we can ask the question: \emph{What is the maximal distinguishing advantage if $\text{Bob}$ and $\text{Charlie}$ are restricted to output $n$ classical bits?}.
We first limit our attention to a special case of this problem such that $n=1$ as well as: \begin{enumerate}
    \item $\distr_1$ outputs two identical Haar random states $\ket{\psi} \ot \ket{\psi}$.
    \item $\distr_2$ outputs two independent Haar random states $\ket{\psi_B} \ot \ket{\psi_C}$.
\end{enumerate}
In both $\distr_1$ and $\distr_2$, the first half of the state will be given to Bob and the second half will be given to Charlie. Note that if we restrict our attention to Bob (or Charlie) alone, then the two cases are perfectly indistinguishable. Therefore, Bob and Charlie need to work collectively in order to achieve a non-trivial distinguishing advantage.
\par For now assume that the pre-shared entanglement consists of some arbitrary number ($r$) of EPR pairs, denoted by $\ket{\epr}^{\ot r}$. Let $M$ and $N$ be the measurements (formally POVM elements) applied by Bob and Charlie, respectively. It is without loss of generality to assume that $M$ and $N$ are projective. In this case, we can write the distinguishing advantage of Bob and Charlie as \begin{align*}
    \tr \bracS{ \brac{M \ot N} \bracC{ \epr^{\ot r} \ot \brac{  \E_{\ket{\psi}} \psi \ot \psi - \E_{\psi_B, \psi_C} \psi_B \ot \psi_C } } }.
\end{align*}
The second expectation equals the maximally mixed state $(\id \ot \id)/d^2$, where $d$ is the dimension of the Haar random states. The first expectation equals the maximally mixed state over the symmetric subspace\footnote{See~\cite{Harrow13} for an introduction to the symmetric subspace.}, which is $O(1/d)$ close to $(\id \ot \id + F)/d^2$, where $F$ is the operator that swaps two registers. Hence, we can approximate the advantage as 
\begin{align*}
    \frac{1}{d^2}\tr \bracS{ \brac{M \ot N} \bracC{ \epr^{\ot r} \ot F } }.
\end{align*}
We examine this expression in terms of tensor network diagrams\footnote{See \cite{Mele24} for an introduction to tensor network diagrams.} \cite{Penrose71}. Up to normalization, the effect of the EPR pairs (followed by trace) is to connect the entangled registers of $M$ and $N$ in reverse (i.e. after partially transposing one of the projectors), whereas the effect of $F$ (followed by trace) is to connect the registers of $M$ and $N$ containing the Haar states. Overall, we observe that the expression above equals \begin{align}
    \frac{1}{2^r d^2} \tr \bracS{ M \cdot N^{\top_P} }, \label{eq:tensor1}
\end{align} 
where $\top_P$ denotes the partial transpose operation. Notice that this is the Hilber-Schmidt inner-product of $M$ and $N^{\top_P}$, hence by Cauchy-Schwarz we can bound it by \begin{align*}
    \frac{1}{2^r d^2} \ellpnorm{M}{2} \cdot \ellpnorm{N^{\top_P}}{2} = \frac{1}{2^r d^2} \ellpnorm{M}{2} \cdot \ellpnorm{N}{2}\le \frac{1}{2^r d^2} \sqrt{2^r d} \cdot \sqrt{2^r d} = \frac{1}{d},
\end{align*}
where we used the fact that $\ellpnorm{M}{2}^2$ equals the rank of $M$. Together with the previous approximation error we had, this gives us a bound of $O(1/d)$. 
\par This bound is in fact tight, which can be seen by a simple attack where Bob and Charlie measure and output the first qubit of their input state. Moreover, this attack is unentangled, leading to the interesting conclusion that entangled attacks are no more powerful than unentangled attacks.

\paragraph{The Case of Many Copies.} Now we generalize our argument to the case when Bob and Charlie each get $t$ copies of their input. Similar to the $t=1$ case, we can write the projection onto the symmetric subspace over $t$ registers as \begin{align*}
    \Pi^t = \frac{1}{t!} \sum_{\sigma \in S_t} P_\sigma,
\end{align*}
where $P_\sigma$ is a register-wise permutation operator. Using this identity for Bob's and Charlie's registers alike, the independent case yields a sum of terms of the form $P_{\sigma_B} \ot P_{\sigma_C}$. On the other hand, the identical case will give us a sum of $P_\sigma$ for $\sigma \in S_{2t}$. We can match the coefficients up to $O(t^2/d)$ error, and thus we can approximate the advantage  as \begin{align*}
    \frac{1}{d^{2t}} \Tr \bracS{ (M \ot N) \, \Big \{ \epr^{\ot r} \ot \sum_{\sigma \in S_{2t}^*} P_\sigma  \Big \} },
\end{align*}
where $S_{2t}^*$ is the set of permutations over $2t$ registers that cannot be expressed as a product of two permutations over $t$ registers. A natural idea would be to bound this expression separately for each $P_\sigma$, but it would yield a factor of $\poly(t!)$ which blows up very fast. Our idea instead is to group permutations based on how far they are from a product permutation. In more detail, we define $S_{2t}^{(s)}$ as the set of permutations $\sigma$ which obeys the identity \begin{align*}
    P_\sigma = \brac{P_{\sigma_1} \ot P_{\sigma_2}} P_{\sigma_s} \brac{P_{\sigma_3}\ot P_{\sigma_4}}
\end{align*}
for some permutations $\sigma_1,\sigma_2,\sigma_3,\sigma_4$ over $t$ registers, where $\sigma_s$ is a fixed permutation that swaps $t$ of Bob's registers with $t$ of Charlie's registers. Indeed, $S_{2t}^{(s)}$ is the set of permutations $\sigma$ that make $s$ \emph{swaps} across Bob's and Charlie's registers. Moreover, by a combinatorial argument, we can compute the average of $P_\sigma$ by averaging the identity above over $\sigma_1,\sigma_2,\sigma_3,\sigma_4$, so that \begin{align*}
    \sum_{\sigma \in S_{2t}^{(s)}} P_\sigma = C_s \brac{\Pi^t_B \ot \Pi^t_C} P_{\sigma_s} \brac{\Pi^t_B \ot \Pi^t_C},
\end{align*}
where $C_s$ is a constant that depends on $s,t$, and $\Pi^t_B,\Pi^t_C$ are projections onto the symmetric subspace over Bob's and Charlie's $t$ registers, respectively. Using the generalization of \cref{eq:tensor1} we show that \begin{align*}
    \frac{1}{d^{2t}} \Tr \bracS{ (M \ot N) \, \Big \{ \epr^{\ot r} \ot \sum_{\sigma \in S_{2t}^{(s)}} P_\sigma  \Big \} } \le \binom{t}{s}^2 s!d^{-s},
\end{align*}
and summing this over $s=1,2,\dots,t$ gives us a bound of $O(t^2/d)$. Keep in mind that $s=0$ is excluded from the sum because it corresponds to product permutations.

\paragraph{The Case of General Entanglement.} Now suppose the entangled state $\ket{\Omega}$ shared between Bob and Charlie is arbitrary. Intuitively, we don't expect this relaxation to help the adversary a lot since the number of EPR pairs above was unbounded, yet it requires a proof to show this. Recall that $\ket{\Omega}$ can be written as \begin{align*}
    \ket{\Omega} =\sum_{i} \sqrt{\lambda_i} \ket{u_i}\ket{v_i}
\end{align*}
for some choice of bases $(u_i), (v_i)$ for the registers of Bob and Charlie, known as the \emph{Schmidt decomposition}. An equivalent way to write this is\footnote{Here we are implicitly assuming that the dimensions of Bob and Charlie's registers both equal the same power of 2. This is merely for convenience and does not affect the analysis.} \begin{align*}
    \ket{\Omega} = (\sigma^{1/2} \ot V) \ket{\epr^{\ot r}}
\end{align*}
for some density matrix $\sigma$, unitary $V$, and integer $r$. The unitary $V$ can be safely ignored by changing the basis of Charlie's projection $N$. Using the cyclicity of trace we can absorb $\sigma^{1/2}$ in $M$ and get a similar expression for the distinguishing advantage as the maximally entangled state, where $M$ is replaced with $(\sigma^{1/2}\ot \id) M (\sigma^{1/2}\ot \id)/2^r$, with $d'$ being the dimension of Bob's share of the entangled state. Following the same analysis as the EPR case results in a bound that scales with $\sqrt{d'}$. In order to get a bound that does not depend on the amount of entanglement we perform a more refined and involved analysis coupled with a more careful application of Cauchy-Schwarz, which yields a bound of $O(t^2/\sqrt{d})$ on the distinguishing advantage. An interesting open question is whether the gap between the EPR and non-EPR cases is inherent.

\subsubsection{Applications}

\paragraph{Unclonable Encryption.}
The search (weak) security of unclonable encryption (UE) was known since its formal introduction \cite{BL20}, yet strong (CPA-style/indistinguishability) security has been an open problem. There are fundamental reasons why this problem has been difficult, including the following: \begin{enumerate}
    \item Because the adversary learns the secret key in the challenge phase of the unclonable security experiment, it is hard to leverage traditional cryptographic tools -- wherein revealing the secret key tantamounts to the compromise of security -- in the construction of unclonable encryption.
    \item There is a lack of straightforward equivalence between unpredictability and indistinguishability in the unclonability setting. The former is used to define CPA-style security and is not transitive, hence unfriendly to hybrid arguments.
    \item Due to the simultaneous nature of the security experiment, extraction techniques that work for a single party often fail against two or more entangled parties.
\end{enumerate}
To elaborate on the third bullet further, one can hope to deploy classical techniques for search-to-decision reductions in this setting. For instance, it has been shown that random oracles can be used to go from weak security to strong security in UE \cite{AKLLZ22,AKL23}. In the plain model, a common classical tool is the Goldreich-Levin extraction technique \cite{GL89}, using which one can try the following folklore construction of UE: \begin{enumerate}
    \item The key consists of $(k,r)$, where $k$ is a key for a weakly secure UE scheme and $r \in \bit^n$ is a random string.
    \item To encrypt a single-bit message $m \in \bit$, sample a random message $x\in \bit^n$, then output encryption of $x$ using $k$, as well as $\inner{r}{x} \oplus m$.
    \item To decrypt, first use $k$ to recover $x$ using the decryption procedure of the weakly secure UE scheme and then recover $m$.
\end{enumerate}
To prove unclonable security of this construction, one needs the \emph{identical-challenge} version of simultaneous Goldreich-Levin, where Bob and Charlie will get the same $r$ as challenge. This is unknown even though the independent-challenge version is known \cite{KT22,AKL23}.
\par Our main insight is to make the string $r$ in the key come in superposition, i.e. from a quantum state $\sum \alpha_r \ket{r}$. Intuitively, if Bob and Charlie were to measure $r$ in the computational basis, then they would effectively receive independent values of $r$, meaning that we can hope to use independent-challenge Goldreich-Levin. Accordingly, we look for a state $\sum \alpha_r \ket{r}$ such that (1) Bob and Charlie cannot simultaneously distinguish whether or not this state has been measured in the computational basis, and (2) the computational basis measurement results in a uniform value of $r$. 
\par Perhaps the most natural candidate for this task is to pick a Haar random state. This allows us to apply our simultaneous Haar indistinguishability result. Nonetheless, there still remain some technical challenges in the application of this concept. To begin with, we need to adapt the construction slightly to incorporate the newly acquired quantumness of $r$. Our solution is as follows: 
 
\begin{enumerate}
    \item The key is partially quantum: \begin{itemize}
        \item The (classical) \emph{encryption key} consists of $(k,x,\widetilde{b})$, where $x$ is a random message and $\widetilde{b} \in \bit$ is a single-bit one-time-pad.
        \item The (quantum) \emph{decryption key} consists of $k$ and a state $\sum \alpha_r\ket{r}\ket{\inner{r}{x} \oplus \widetilde{b}}$, where $\sum \alpha_r\ket{r}$ is a Haar random state.
    \end{itemize} 
    \item To encrypt a single-bit message $m \in \bit$, output encryption of $x$ using $k$, as well as $\widetilde{b} \oplus m$.
    \item To decrypt, first use $k$ to recover $x$, then coherently recover $\widetilde{b}$ followed by $m$.
\end{enumerate}
Using the simultaneous Haar indistinguishability, we can show that our construction is secure via the hybrid method. Note that we can do this because our notion of simultaneous-state-indistinguishability is strong enough that it is amenable to the use of hybrids.
\par In more detail, we reach an indistinguishable hybrid experiment where the  Bob and Charlie get keys which use independently generated Haar random states. Equivalently, Bob gets $(r, \inner{r}{x} \oplus \widetilde{b})$ and Charlie gets $(r', \inner{r'}{x}\oplus \widetilde{b})$ for independent $r,r'$.
\par Next, we move to a hybrid which is the weak security (i.e. search security) experiment of the underlying unclonable encryption scheme, so that Bob and Charlie need to output $x$ each given the key $k$. Unfortunately, even though $r$ and $r'$ are independent in the previous hybrid, independent-challenge simultaneous Goldreich-Levin \cite{KT22,AKL23} is insufficient due to the correlation between the bits $b=\inner{r}{x} \oplus \widetilde{b}$ and $b'=\inner{r'}{x} \oplus \widetilde{b}$, namely $b \oplus b' = \inner{r}{x} \oplus \inner{r'}{x}$. To overcome this issue, we prove exactly what we need, which we call \emph{correlated simultaneous quantum Goldreich-Levin\footnote{As a side note, this lemma resolves an open question in \cite{KT22}, implying that their construction achieves a more desirable notion of security.}} \footnote{Previously, the work of~\cite{AB23} explicitly stated the correlated Goldreich-Levin problem over large finite fields as a conjecture.} (\Cref{lem:gl_cor}), which can be summarized as follows:
\begin{quote}
    \textsc{Correlated Goldreich-Levin:} Suppose that Bob is given input $(r,b)$ and Charlie is given input $(r',b')$, where $r,r'$ are independent strings and $b,b'$ are uniform bits satisfying the correlation $b \oplus b' = \inner{r}{x} \oplus \inner{r'}{x}$. If (Bob, Charlie) can output $(\inner{r}{x}, \inner{r'}{x})$ with probability $1/2 + 1/\poly$, then there is an extractor $(\textrm{ExtBob},\textrm{ExtCharlie})$ that extracts $(x,x)$ from (Bob, Charlie) with probability $1/\poly$.
\end{quote}
To prove this lemma we first tackle the correlation between $b$ and $b'$. Consider $\bobx$ who takes as input $r$, samples $b$ himself uniformly, and runs Bob on input $(r,b)$ to obtain $b_B$; similarly consider $\charliex$ who takes $r'$ as input, samples $b'$ and runs Charlie on input $(r',b')$ to obtain $b_C$. Now that the input bits $b,b'$ are uncorrelated, ($\bobx, \charliex$), who output $(b_B, b_C)$, are expected to have worse success probability than (Bob, Charlie). However, we can in turn relax the success criterion for ($\bobx, \charliex$) to merely output bits $(b_B,b_C)$ that satisfy $b_B \oplus b_C = \inner{r}{x} \oplus \inner{r'}{x}$ in order to counteract this lack of correlation. In other words, now ($\bobx, \charliex$) are additionally allowed to be \emph{both incorrect}. Indeed, we show that the success probability of ($\bobx, \charliex$) in this case is at least that of (Bob, Charlie), i.e. $1/2 + 1/\poly$. To show this fact, we define $E$ as the event that $b \oplus b' = \inner{r}{x} \oplus \inner{r'}{x}$. Conditioned on $E$, Bob and Charlie will output $(\inner{r}{x}, \inner{r'}{x})$ with probability $1/2 + 1/\poly$ by our assumption. In addition, the event $E$ is independent of Bob's (or Charlie's) marginal output due to the fact that the players' correlation satisfies no-signalling. To see this, notice that the bits $b$ and $b'$ can each independently control the event $E$. We utilize this important observation to show the desired result.
Another way to interpret this reduction is as follows: the correlation that (Bob, Charlie) require in order to output $(\inner{r}{x}, \inner{r'}{x})$ appears as a correlation in the output of ($\bobx, \charliex$), who take uncorrelated bits as input.

\par After this reduction, it seems that we still cannot use the original independent-challenge simultaneous Goldreich-Levin because of the relaxed success criterion above. Luckily, by examining the proof of \cite{AKL23} we see that this condition is sufficient without additional work for the existence of (ExtBob, ExtCharlie) who can extract $(x,x)$ simultaneously. 

\paragraph{Many-Copy Security.} For $t$-copy security, where Bob and Charlie get $t$ copies of the decryption key in the unclonable security experiment, we need $t$ to be at most linear in $n$, for otherwise Bob and Charlie can learn $x$ using Gaussian elimination. In the proof, we similarly reach a hybrid where the Haar random states given to Bob and Charlie are independent. Then, we need an extra step where we switch to a hybrid in which Bob gets as input $(r_i, \inner{r_i}{x} \oplus \widetilde{b})$ for independent samples $r_1,\dots,r_t$ (instead of $(r_1,\dots,r_t)$ being generated from $t$ copies of a Haar random state) and similarly Charlie gets $(r_i', \inner{r_i'}{x} \oplus \widetilde{b})$ for independently generated $r_1',\dots,r_t'$. We show that the success probability of Bob and Charlie does not decrease from this change. To show this, we argue that that given $(r_i, \inner{r_i}{x} \oplus \widetilde{b})$ Bob can prepare $$ \brac{\sum_{r} \alpha_r \ket{r}\ket{\inner{r}{x} \oplus \widetilde{b}} }^{\otimes t} $$
where $\ket{\varphi} = \sum_r \alpha_r \ket{r}$ is a Haar random state. In the expression above, $\ket{r}$ corresponds to the register that holds Goldreich-Levin samples, which are generated from $t$ copies of a Haar random state rather than $t$ independent samples. Recall that $\varphi^{\ot t}$ can be written as a random vector in the \emph{type basis}. Using this fact, Bob can coherently apply a random permutation $\sigma \in S_t$ to the values $\ket{r_i, \inner{r_i}{x} \oplus \widetilde{b}}$, after which he can uncompute the permutation $\sigma$. We can argue similarly for Charlie, since their inputs originate from independent Haar distributions. This argument in fact requires the strings $r_i$ to be distinct, which fortunately holds with high probability by the birthday bound. 
\par Note that the input above given to Bob can be thought of as $(r, \inner{r}{x} \oplus \widetilde{b})$ alongside $t-1$ random samples of $(r_i, \inner{r_i}{x})$. In the final step, we apply our correlated Goldreich-Levin result in the presence of this extra information to reach the search security experiment for the weakly secure UE scheme. In this experiment, the extra information can be guessed by Bob and Charlie, hence if $t$ is bounded by a linear function of $n$ the security still holds.

\paragraph{Single-Decryptor Encryption.}
Single-decryptor encryption is a primitive that closely resembles unclonable encryption. It is an encryption scheme in which the decryption key is unclonable. In the security experiment, Alice tries to clone a quantum decryption key and split it between Bob and Charlie, who then try to decrypt a ciphertext they receive using their shares of the key. Depending on the correlation of these ciphertexts, one can define \emph{identical-challenge} or \emph{independent-challenge} security. We adapt our construction of unclonable encryption with quantum decryption keys to construct single-decryptor encryption with quantum ciphertexts. Our construction can be summarized as follows: \begin{enumerate}
    \item The encryption key consists of a key $k$ as well as a random message $x$ for a weakly secure UE scheme. The quantum decryption key contains $k$ and encryption of $x$ using $k$.
  
    \item To encrypt a one-bit message $m$, output $k$ as well as $\sum \alpha_r \ket{r} \ket{\inner{r}{x} \oplus m}$, where $\sum \alpha_r \ket{r}$ is a Haar random state.
    \item To decrypt, first recover $x$ and then coherently compute $m$.
\end{enumerate}

We can show that this construction is secure if Bob and Charlie are given $t$ copies of the same ciphertext for $t=O(|x|)$. The proof is nearly identical to the security proof of our UE construction above. For more clear exposition, in the technical sections we first present our construction of single-decryptor encryption (\Cref{sec:sde}), followed by that of unclonable encryption (\Cref{sec:ue}).  

\paragraph{Classical-Leakage-Resilient Secret-Sharing.}
As another application of simultaneous Haar indistinguishability, we construct a 2-out-of-$n$ quantum secret-sharing scheme for a single-bit classical message. The construction is as follows: \begin{enumerate}
    \item The shares of bit $b=0$ are identical Haar random states $\ket{\psi}$ and the shares of $b=1$ are independent Haar random states $\ket{\psi_i}$. In addition, we give $t$ copies of their share to all parties $1,2,\dots,n$.
    \item Any two parties $i,j$ can recover the message by applying $t$ SWAP tests between their secrets. If $m=0$, then all the tests will succeed. On the other hand, if $m=1$, then with high probability the independent Haar random states held by $i$ and $j$ will be almost orthogonal, therefore the number of SWAP tests that pass will be concentrated near $t/2$ by a Chernoff bound.
\end{enumerate}
Formally, we show that $m$ remains hidden in the presence of $\ell$-bits of classical leakage from each party. This amounts to showing a many-party and many-bit generalization of simultaneous Haar indistinguishability (SHI), that is, either all parties get (copies of) the same Haar random state or they get (copies of) independent Haar random states. Firstly, we go from 1-bit SHI to $\ell$-bit SHI for 2 parties. This can be achieved by a union bound, which incurs a multiplicative loss of $2^{2\ell}$ in security. And then we show an equivalence between SHI in the cases of (1) $2$ parties each getting $O(nt)$ copies and (2) $O(n)$ parties each getting $t$ copies. This can be seen as distributing a fixed number of Haar random states among more parties. In the proof we use $O(\log n)$ hybrids, doubling the number of parties at each hybrid. We show by an additional hybrid argument that we incur a multiplicative loss of $O(1)$ at each step, hence a multiplicative loss of $O(n)$ in total after $O(\log n)$ steps. Putting everything together, we show that we can allow $\ell = O(\log d/n)$ bits of leakage from each party, where $\log d$ is the number of qubits of each copy of a share. The number of copies ($t$) given to each party can be an arbitrary polynomial, which lets us amplify the correctness of the scheme. An interesting open question is to get rid of the exponential dependence on the number of bits leaked ($\ell n$), which would imply that our construction tolerates unbounded polynomial leakage.
\newcommand{\ellonenorm}[1]{\ellpnorm{#1}{1}}
\newcommand{\sphere}{{\cal S}}
\newcommand{\innerhs}[2]{\inner{#1}{#2}_{\sf HS}}
\newcommand{\SymmSubspace}{\cS_\mathsf{sym}}
\newcommand{\opnorm}[1]{\left\|#1 \right\|_{\mathsf{op}}}
\newcommand{\type}{\mathsf{type}}
\newcommand{\distinct}{\Lambda}

\section{Notation and Preliminaries}

\paragraph{Notation.} We write $\log := \log_e$ to denote the natural logarithm. We write $[n] := \bracC{1,\dots,n}$. We use the notation $\inner{\cdot}{\cdot} : \bit^n \times \bit^n \to \bit$ to denote the inner product over $\F_2^n$, i.e. for classical strings $x,y \in \bit^n$ we have $\inner{x}{y} := \sum_{i=1}^n x_iy_i \pmod{2}$. We denote the set of $d$-dimensional pure quantum states by $\sphere(\C^d) := \bracC{ \ket{\psi} \in \C^d \; : \; \braket{\psi}{\psi} = 1 }$. We sometimes write $\psi := \ketbraX{\psi}$ for simplicity. For a complex matrix (or operator) $A$, we write $A^\top$ to denote its transpose and $\overline{A}$ to denote its entry-wise complex conjugation, both with respect to the computational basis.

\paragraph{Total Variation Distance.} For random variables $x,x'$ over a set $X$, their total variation distance is defined as $$ \tvdist{x}{x'} := \max_{S \subseteq X} \pr{ x\in S } - \pr{x' \in S} = \frac{1}{2} \sum_{x \in X} \abs{ \pr{ x\in S } - \pr{x' \in S} }. $$

\paragraph{Quantum Computing.} Quantum registers are denoted using the font $\reg{X}$. The dimension of the Hilbert space associated with $\reg X$ is denoted by $\dim(\reg X)$. A quantum polynomial-time (QPT) algorithm $\alice$ is one which applies a sequence of polynomially many basic operations (universal quantum gate, qubit measurement, initializing qubit to $\ket{0}$).
\par A \emph{pseudo-deterministic} quantum algorithm $\alice$ takes as input a quantum state $\ket{\psi}$ and classical randomness $r$. It always outputs the same quantum state $\ket{\varphi}$ for fixed $(\ket{\psi}, r)$.

\paragraph{Haar Measure.} We denote by $\unitarygp_d$ the unitary group over $\C^d$. $\haar(\unitarygp_d)$ denotes the unique Haar measure over $\unitarygp_d$, and $\haar_d$ denotes the uniform spherical measure induced by $\haar(\unitarygp_d)$ on $\C^{d}$. States sampled from $\haar_d$ are referred to as \emph{Haar random states}. 

\paragraph{Non-Local Adversaries.}

\noindent A \emph{non-local adversary} is a tuple $\adversary = (\bob,\charlie,\rho_\reg{BC})$, where $\bob$, $\charlie$ are physically separated quantum parties and $\rho_\reg{BC}$ is a bipartite state shared between them. In addition, $\alice$ takes as input a bipartite state $(\sigma_\reg{B' C'})$, simultaneously computes $\bob$ on registers $\reg{B,B'}$ and computes $\charlie$ on registers $\reg{C,C'}$. Furthermore, we say that $\adversary$ is efficient if $\bob$ and $\charlie$ are quantum polynomial-time (QPT) algorithms. If part of the input $\sigma_\reg{B' C'}$ is classical, we write $\sigma_\reg{B' C'} = \brac{x, x', \sigma'_\reg{B' C'} }$, where $x$ is given to $\bob$ and $x'$ is given to $\charlie$. 

\par Below we list a slightly stronger version of the Simultaneous Quantum Goldreich-Levin lemma proved in \cite{AKL23}. The proof given in \cite{AKL23} works as is for the statement below, but we rewrite it in full in \Cref{sec:additional_proofs} for the sake of completeness.

\begin{lemma}[Simultaneous Quantum Goldreich-Levin]
\label{lem:quantumGL}
Let $x \in \bit^n$ be a random variable. Suppose a non-local adversary $\adversary = (\bob,\charlie,\rho)$ given input $(r,r',\sigma)$, where $r,r' \in \bit^n$ are i.i.d. uniform strings, can output bits $(b,b')$ satisfying $b \oplus b' = \inner{r}{x} \oplus \inner{r'}{x}$ with probability at least $1/2 + \varepsilon$. Then, there exists a non-local (extractor) adversary $\aliceprime = (\bobprime, \charlieprime, \rho)$ which can output $(x,x)$ with probability $\poly(\varepsilon)$ given input $\sigma$. Furthermore, $\bobprime$ and $\charlieprime$ run $\bob$ and $\charlie$ as subprotocols once, respectively.
\end{lemma}
\par A \emph{$q$-party non-local adversary} $\alice=(\alice_1,\dots,\alice_q, \rho)$ is defined similarly, where $\rho$ is a quantum state shared across $\alice_1,\dots,\alice_q$. A $2$-party non-local adversary is simply a non-local adversary.

\par A \emph{cloning adversary} is defined similarly to a non-local adversary, with an additional splitting algorithm at the start. Formally, it is a tuple $\abc$ of quantum algorithms, where $\alice$ (on behalf of $\abc$) takes as input a quantum state $\rho$ and outputs a bipartite state $\rho_\reg{BC}$. $(\bob, \charlie, \rho_\reg{BC})$ then acts as a non-local adversary, taking an additional bipartite state $\sigma_\reg{B'C'}$ (referred to as the challenge) as input.

\paragraph{Matrix Norms.} 

The \emph{Frobenius norm}, or the \emph{Hilbert-Schmidt norm} of a complex matrix $A = \brac{A_{i,j}}_{i,j}$ is defined as
\begin{align*}
    \ellpnorm{A}{2} := \brac{ \sum_{i,j} \abs{A_{i,j}}^2}^{1/2} = \sqrt{\tr \brac{ A^\dagger A } }.
\end{align*}
$\|\cdot \|_2$ is induced by the \emph{Hilbert-Schmidt inner-product}, defined as \begin{align*}
    \innerhs{A}{B} := \sqrt{\tr(A^\dagger B)}.
\end{align*}
The \emph{Loewner order} is a partial order over matrices, where $X \le Y$ if and only if $Y-X$ is positive semidefinite. For $X,Y \ge 0$, we have $\tr(XY) = \tr(Y^{1/2}XY^{1/2}) \ge 0$, which implies the following lemma.
\begin{lemma} \label{lem:monotone}
    Let $0 \le M \le N$ be operators over $\reg{XY}$, then $\ellpnorm{\tr_{\reg Y} M }{2} \le \ellpnorm{\tr_{\reg Y} N }{2}$.
\end{lemma}
\begin{proof}
    $0 \le \tr_{\reg Y} M \le \tr_{\reg Y} N$, hence $\ellpnorm{\tr_{\reg Y} M }{2}^2 = \tr \brac{\tr_{\reg Y} M}^2 \le \tr \brac{\tr_{\reg Y} M}\brac{\tr_{\reg Y} N} \le \tr \brac{\tr_{\reg Y} N}^2 = \ellpnorm{\tr_{\reg Y} N }{2}^2$.
\end{proof}

\noindent The \emph{operator norm} of a matrix $A$, denoted by $\opnorm{A}$, is defined as the largest singular value of $A$. 

\paragraph{Symmetric Subspace.}

Let $\reg{A_1},\dots,\reg{A_t}$ be quantum registers where each $\reg{A_j}$ corresponds to a copy of $\C^d$. We define the symmetric subspace over $\reg{A_1},\dots,\reg{A_t}$ as \begin{align*}
    \SymmSubspace(\reg{A_1}\dots\reg{A_t}) := \mathsf{span} \bracC{\ket{\psi}^{\otimes t}\; : \; \ket{\psi} \in \C^d}.
\end{align*}
Below we list some well-known properties of symmetric subspaces\footnote{We refer the reader to \cite{Harrow13} for a detailed discussion of symmetric subspaces.}. We have \begin{align*}
    \dim \SymmSubspace(\reg{A_1}\dots\reg{A_t}) = \binom{d+t-1}{t}.
\end{align*}
The expectation of $t$ copies of a Haar state is maximally mixed in the symmetric subspace, i.e. \begin{align}
    \E_{\ket{\psi} \from \haar_d} \psi^{\otimes t} = \frac{\SymmSub(\reg{A_1\dots A_t})}{\dim \SymmSubspace(\reg{A_1}\dots\reg{A_t})} = \frac{\SymmSub(\reg{A_1\dots A_t})}{\binom{d+t-1}{t}}, \label{eq:sym_sub}
\end{align}
where $\SymmSub(\reg{A_1\dots A_t})$ denotes the projector onto $\SymmSubspace(\reg{A_1}\dots\reg{A_t}).$

\par We will denote the symmetric group over $t$ elements by $S_t$. For a permutation $\sigma \in S_t$ and $t$ quantum registers of equal dimension, we define the (register-wise) permutation operator $P_\sigma : \cH^{\otimes t} \to \cH^{\otimes t}$ as $$ P_\sigma := \sum_{x_1,\dots,x_t} \ket{x_{\sigma(1)},\dots,x_{\sigma(t)}}\bra{x_1\dots x_t}. $$

We have the following useful identity:
\begin{lemma} \label{lem:sum_perm}
    \begin{align*}
        \SymmSub(\reg{A_1\dots A_t}) = \frac{1}{t!} \sum_{\sigma \in S_t} P_\sigma,
    \end{align*}
        where $P_\sigma$ permutes $\reg{A_1, \dots ,A_t}$.
\end{lemma}

Given a maximally mixed state over the symmetric subspace, performing a partial trace leaves the state maximally mixed over the symmetric subspace on the remaining registers. Formally, we have the following lemma.

\begin{lemma} \label{lem:sym_part_tr}
Let $\dim(\reg{A_i}) = d$ for $i \in [t]$, then \begin{align*}
    \tr_\reg{A_{s+1}\dots A_t} \frac{\SymmSub(\reg{A_1\dots A_t})}{\binom{d+t-1}{t}} =  \frac{\SymmSub(\reg{A_1\dots A_s})}{\binom{d+s-1}{s}}.
\end{align*}
\end{lemma}

\begin{proof} By \cref{eq:sym_sub} and linearity of expectation, we have
    \begin{align*}
        \tr_\reg{A_{s+1}\dots A_t} \frac{\SymmSub(\reg{A_1\dots A_t})}{\binom{d+t-1}{t}} = \tr_\reg{A_{s+1}\dots A_t} \E_{\ket{\psi} \from \haar_d} \psi^{\ot t} = \E_{\ket{\psi} \from \haar_d} \psi^{\ot s} = \frac{\SymmSub(\reg{A_1\dots A_s})}{\binom{d+s-1}{s}}.
    \end{align*}
\end{proof}

\par Denote the computational basis for $\reg{A_i}$ with strings $x \in X$, with $|X| = d$. Given $(x_1,\dots,x_t) \in X^t$, we define the \emph{type state} as \begin{align*}
    \ket{\type(x_1,\dots,x_t)} := \frac{1}{t!}\sum_{\sigma \in S_t} P_\sigma \ket{x_1\dots x_t} = \SymmSub(\reg{A_1 \dots A_t}) \ket{x_1, \dots, x_t}
\end{align*}
The type states form an orthonormal basis of the symmetric subspace, with the caution that the map $(x_1, \dots, x_t) \to \ket{\type(x_1,\dots,x_t)}$ is not injective. Let $\distinct_{t,d}$ be the set of $t$-tuples satisfying $x_1 \ne \dots \ne x_t$. We have the following lemma, which follows by the birthday bound:
\begin{lemma} \label{lem:distinct_type}
    \begin{align*}
         \tracedist{\E_{(x_1,\dots,x_t) \from \distinct(t,d)} \ketbraX{\type(x_1,\dots,x_t)}} {\frac{\SymmSub(\reg{A_1\dots A_t})}{\binom{d+t-1}{t}}} \le O\brac{\frac{t^2}{d}}.
    \end{align*}
\end{lemma}

\paragraph{Unclonable Encryption with Weak Unclonable Security.} A \emph{one-time unclonable encryption scheme} is a tuple $(\gen_\ue, \enc_\ue, \dec_\ue)$ of algorithms alongside a message space $\cM \subseteq \bit^{n}$, where $n = \poly(\secparam)$  with the following syntax: \begin{itemize}
    \item $\gen_\ue(1^\secparam)$ takes as input a security parameter and outputs a classical key $k$.
    \item $\enc_\ue(k, m)$ takes as input a key $k$ and a message $m \in \cM$; it outputs a quantum ciphertext $\ket{\ct}$.
    \item $\dec_\ue(k, \ket{\ct})$ takes as input a key $k$ and a quantum ciphertext $\ket{\ct}$; it outputs a classical message $m$.
\end{itemize}

\begin{definition}[Correctness] \label{def:ue_correctness}
    An unclonable encryption scheme $(\gen_\ue, \enc_\ue, \dec_\ue)$ is \emph{correct} if for any security parameter $\secparam$ and any message $m$ we have \begin{align*}
        \pr{ m' = m \; : \; \substack{ k \from \gen(1^\secparam) \\ \qcipher \from \enc(k, m) \\ m' \from \dec(k, \qcipher)} } \ge 1 - \negl(\secparam).
    \end{align*}
\end{definition}

\begin{definition}[$\eps$-Weak Unclonable Security]\label{def:ue_sec_weak}
A one-time unclonable encryption scheme $(\gen_\ue, \enc_\ue, \dec_\ue)$ has \emph{$\eps$-weak unclonable security} if for any cloning adversary $\abc$ and any pair of messages $m_0, m_1 \in \cM$ we have \begin{align*}
        \pr{ m_\bob = m_\charlie = m \; : \; \substack{k \from \gen(1^\secparam)\vspace{-.15cm} \\ m \uniform \cM, \quad \qcipher \from \enc(k, m) \\ \rho_\reg{BC} \from \alice(\qcipher) \\  m_\bob \from \bob(k, \rho_\reg{B}), \quad m_\charlie \from \charlie(k, \rho_\reg{C}) } } \le \eps,
    \end{align*}
    where $\rho_\reg{E}$ denotes the $\reg{E}$ register of the bipartite state $\rho_\reg{BC}$ for $\reg{E} \in \{\reg{B},\reg{C}\}$.
\end{definition}

\par Broadbent and Lord~\cite{BL20} showed how to achieve unclonable encryption with $\eps$-weak unclonable security for $\eps < 0.86^\secparam$ and $n = \secparam$ using Wiesner states.
\newcommand{\plus}{\texttt{+}}
\newcommand{\minus}{\texttt{-}}
\newcommand{\rk}{\mathsf{rk}}

\section{Simultaneous State Indistinguishability}

\paragraph{Terminology.}
Below, $\distr$ represents a probability distribution over pure states. Particularly, we will consider bipartite pure states $\ket{\psi}_\reg{BC}$ and non-local adversaries $\alice = (\bob, \charlie, \rho_{\reg{BC}})$ as distinguishers, where the $\reg{B}$ register will be given to $\bob$ and the $\reg{C}$ register will be given to $\charlie$.

\subsection{Definitions}
\label{sec:ssi_definition}

\begin{remark}
    A distribution $\distr$ has a unique representation as a quantum mixed state $\rho_\distr$, whereas a quantum mixed state can represent many distributions. Accordingly, we can apply our results to mixed states $\rho$ for which we can find an appropriate distribution $\distr'$ that satisfies $\rho_{\distr'} = \rho$.
\end{remark}

\begin{definition}[Simultaneous State Indistinguishability] 
We say that two distributions $\distr_1$ and $\distr_2$ are \emph{$\eps$-simultaneous state indistinguishable} ($\eps$-SSI) against a non-local adversary $\adversary=(\bob,\charlie,\rho_{\reg{BC}})$ if the following holds for every pair of bits $b_1,b_2$: 
$$\absbig{ \prob\left[(b_1,b_2) \leftarrow \adversary({\psi}_{\reg{BC}}) \; : \; \ket{\psi}_{\reg{BC}} \from \distr_1 \right] - \prob\left[(b_1,b_2) \leftarrow \adversary({\psi}_{\reg{BC}}) \; : \; \ket{\psi}_{\reg{BC}} \from \distr_2 \right] } \leq \eps$$

Here, $\bob$ gets the register $\reg{B}$ of $\psi$ and $\charlie$ gets the register $\reg{C}$. 
\end{definition} 

\noindent Of interest to us is the case when $\distr_1 = \distr^\id$ and $\distr_2 = \distr^\ind$, where we define $\distr^\id$ and $\distr^\ind$ as follows\footnote{Here $\id$ stands for \emph{identical} and $\ind$ stands for \emph{independent}.}:
\begin{enumerate}
    \item[] $(\distr^\id):$ Sample $\ket{\psi} \from \distr$ and output $(\ket{\psi}_{\reg{B}} \otimes \ket{\psi}_{\reg{C}})$
    \item[] $(\distr^\ind):$ Sample $\ket{\psi}, \ket{\psi'} \from \distr$ and output $(\ket{\psi}_{\reg{B}} \otimes \ket{\psi'}_{\reg{C}})$, i.e. $\distr^\ind = \distr \times \distr$.
\end{enumerate}
In this case, we refer to the above notion as $(\eps,\distr)$-simultaneous state indistinguishability ($(\eps,\distr)$-SSI). Note that up to a constant factor this is equivalent to the output distributions of $\alice$ with respect to $\distr_1, \distr_2$ having total variation distance $\eps$. Also note that we can fix the bits $b_1,b_2$ without loss of generality.

\begin{definition}[$(\eps, \distr)$-SSI]
    We say that \emph{($(\eps,\distr)$-SSI)} holds if $\distr^\id$ and $\distr^\ind$ defined above are $\eps$-SSI against all non-local adversaries $(\bob, \charlie, \rho)$. 
   
\end{definition}

\begin{remark}
When it is clear from the context, we omit the mention of the registers $\reg{B}$ and $\reg{C}$. 
\end{remark}

\paragraph{SSI as a metric.}

This notion defines a metric over bipartite states, namely $\nldist{\rho_0}{\rho_1}$ is the smallest value of $\varepsilon$ such that $\rho_0$ and $\rho_1$ are $\varepsilon$ simultaneously indistinguishable. Equivalently, \begin{align}
    \nldist{\rho_0}{\rho_1} := \sup_{(\bob, \charlie, \rho)} \abs{ \pr{ (1,1) \from (\bob \otimes \charlie)(\rho \otimes \rho_0) } - \pr{ (1,1) \from (\bob \otimes \charlie)( \rho \otimes \rho_1)} }. \label{eq:nldist}
\end{align}
It is easy to see that this is a valid metric\footnote{Technically, a pseudometric since two different states may have distance 0.}.

Since any binary measurement is a linear combination of projective measurements by the Spectral Theorem, we have the following fact. \begin{lemma} \label{lem:ssi_proj}
    $\nldist{\cdot}{\cdot}$ can be equivalently defined by restricting $\bob, \charlie$ to be projective measurements.
\end{lemma}
\begin{proof}
    For any $\bob, \charlie$ the expression inside the supremum in \cref{eq:nldist} can be written as a convex combination of the same expression for projective $\bob_j, \charlie_j$ by looking at the spectral decomposition of the positive semi-definite operators $\bob, \charlie$. The lemma follows by a simple triangle inequality.
\end{proof}

\paragraph{Extending to Many-Bit Output.}
We can extend the definition of SSI such that the non-local adversary can output many bits in each register. We then require that the total variation distance between the collective outputs of $\bob$ and $\charlie$ for the two distributions is small.
\begin{definition}[$(\eps, \distr, n)$-SSI] \label{def:ssi_multibit}
    We say that $(\eps,\distr, n)$-SSI holds if for any nonlocal adversary $\alice = (\bob, \charlie, \rho)$ and any $S \subseteq \bit^{n} \times \bit^n$ we have \begin{align*}
        \abs{ \pr{ (x_1,x_2) \in S \; : \; \substack{\ket{\psi} \from \distr \\ (x_1,x_2) \from \alice({\psi} \otimes {\psi})} } - \pr{ (x_1,x_2) \in S \; : \; \substack{\ket{\psi}, \ket{\psi'} \from \distr \\ (x_1,x_2) \from \alice({\psi} \otimes {\psi'})} } } \le \eps
    \end{align*}
\end{definition}
Note that $(\eps, \distr, 2)$-SSI and $(\eps,\distr)$-SSI are equivalent up to a constant factor on $\eps$. There is a straightforward relation between $(\eps, \distr)$-SSI and $(\eps, \distr, n)$-SSI using the union bound, which we formally state below.

\begin{lemma}   \label{lem:ssi_union_bound}
    Suppose $(\eps, \distr)$-SSI holds, then $(2^{2n-1}\eps, \distr, n)$-SSI holds for all $n \in \N$.
\end{lemma}
\begin{proof}
    For any non-local adversary $\alice = (\bob, \charlie, \rho)$ and any $y,y' \in \bit^n$, we have \begin{align*}
        \abs{\pr{(y,y') \from \alice({\psi} \ot {\psi}) \; : \; \ket{\psi} \from \distr} - \pr{(y,y') \from \alice({\psi} \ot {\psi'}) \; : \; \ket{\psi}, \ket{\psi'} \from \distr}} \le \eps
    \end{align*}
    by $(\eps, \distr)$-SSI. This is because $\bob$ ($\charlie$) can associate $y$ (respectively, $y'$) with 0 and every other string with 1. By summing over all $y,y'$ and dividing by $2$ we get the desired result.
\end{proof}

\paragraph{Extending to Many Copies.} For a distribution $\distr$, denote by $\distr^t$ the distribution that samples $\ket{\psi} \from \distr$ and outputs $\ket{\psi}^{\otimes t}$. Then, we can consider $(\eps, \distr^t)$-SSI or $(\eps, \distr^t, n)$-SSI as extensions where $\bob$ and $\charlie$ each get $t$ copies of their respective inputs.

\paragraph{Extending to Many Parties.} We can consider as the distinguisher a $q$-party non-local adversary $\alice = (\alice_1,\dots,\alice_q, \rho_{\reg{A_1}\cdots\reg{A_q}})$. By giving every party $t$ copies of a quantum state generated either identically or independently we can generalize \Cref{def:ssi_multibit}:

\begin{definition}[($\eps, \distr, n$)-SSI against $q$ Parties]  \label{def:ssi_manyparty}
    We say that $(\eps,\distr, n)$-SSI holds \emph{against $q$ parties} if for any $q$-party nonlocal adversary $\alice = (\alice_1,\dots,\alice_q, \rho_{\reg{A_1}\cdots\reg{A_q}})$ and any $S \subseteq \bit^{nq}$ we have \begin{align*}
        \abs{ \pr{ (x_1,\dots,x_q) \in S \; : \; \substack{\ket{\psi} \from \distr \\ (x_1,\dots,x_q) \from \alice({\psi}_{\reg{A_1}} \otimes \cdots \otimes \psi_{\reg{A_q}})} } - \pr{ (x_1,\dots,x_q) \in S \; : \; \substack{\ket{\psi_1}, \dots, \ket{\psi_q} \from \distr \\ (x_1,\dots,x_q) \from \alice({{(\psi_1)}}_{\reg{A_1}} \ot \dots \ot {(\psi_q)}_{\reg{A_q}})} } } \le \eps
    \end{align*}
\end{definition}

\noindent In general, SSI against many parties is weaker than regular SSI with the same number of copies and the same total output length, which we state formally below.

\begin{lemma}   \label{lem:ssi_many_party}
    Suppose that $(\eps, \distr^{qt}, qn)$-SSI holds (against 2 parties), then $(2q\eps, \distr^t, n)$-SSI holds against $q$ parties.
\end{lemma}

\begin{proof}
    First we observe the monotonicity of SSI, that is, for $t \ge t', n \ge n'$, and $q \ge q'$, $(\eps, \distr^t, n)$-SSI against $q$ parties implies $(\eps, \distr^{t'}, n')$-SSI against $q'$ parties. Now onto proving the lemma, the case $q=1$ is trivial. Let $r \ge 0$ such that $q' := 2^r < q \le 2^{r+1}$. Then $(\eps, \distr^{2^r t}, 2^r n)$-SSI holds against $2$ parties. First, by induction on $0 \le j \le r$, we will show that $(\eps_j, \distr^{2^{r-j}t}, 2^{r-j} n)$-SSI holds against $2^{j+1}$ parties, where $\eps_j = (2^{j+1}-1)\eps$. The base case is true by assumption. Suppose it is true for some $j \le r-1$, then we will show that $(\eps_{j+1}, \distr^{2^{r-j-1}t}, 2^{r-j-1}n)$-SSI holds against $2^{j+2}$ parties. Consider a $2^{j+2}$-party non-local adversary $\alice = (\alice_1,\dots,\alice_{2^{j+2}}, \rho)$. We consider the following hybrids: \begin{itemize}
        \item {\bf Hybrid 0:} In this hybrid, $\alice_i$ gets as input ${\ket{\psi}^{\ot 2^{r-j-1}t}}$ for $i \in [2^{j+2}]$, where $\ket{\psi} \from \distr$.
        \item {\bf Hybrid 1:} In this hybrid, $\alice_i$ gets as input ${\ket{\psi}^{\ot 2^{r-j-1}t}}$ for $i \in [2^{j+1}]$ and $\alice_{i'}$ gets as input ${\ket{\psi'}^{\ot 2^{r-j-1}t}}$ for $i' \in [2^{j+2}]\setminus [2^{j+1}]$, where $\ket{\psi}, \ket{\psi'} \from \distr$.
        \item {\bf Hybrid 2:} In this hybrid, $\alice_i$ gets as input $\ket{\psi_i}^{\ot 2^{r-j-1}t}$ for $i \in [2^{j+1}]$ and $\alice_{i'}$ gets as input ${\ket{\psi'}^{\ot 2^{r-j-1}t}}$ for $i' \in [2^{j+2}]\setminus [2^{j+1}]$, where $\ket{\psi_1}, \dots, \ket{\psi_{2^{j+1}}}, \ket{\psi'} \from \distr$.
        \item {\bf Hybrid 3:} In this hybrid, $\alice_i$ gets as input $\ket{\psi_i}^{\ot 2^{r-j-1}t}$ for $i \in [2^{j+1}]$ and $\alice_{i'}$ gets as input $\ket{\psi'_{i'-2^{j+1}}}^{\ot 2^{r-j-1}t}$ for $i' \in [2^{j+2}]\setminus [2^{j+1}]$, where $\ket{\psi_1}, \dots, \ket{\psi_{2^{j+1}}}, \ket{\psi_{1}'}, \dots, \ket{\psi'_{2^{j+1}}} \from \distr$.
    \end{itemize}
    Let $y_0,y_1,y_2,y_3$ denote the output distributions of $\alice$ in {\bf Hybrid 0} to {\bf Hybrid 3}, respectively. We will split the proof into the following claims:
    \begin{claim}   \label{clm:tv_01}
        $\tvdist{y_0}{y_1} \le \eps$.
    \end{claim}
    \begin{proof}
        Define a (2-party) non-local adversary $\aliceprime = (\bobprime, \charlieprime, \rho)$ as follows: $\bobprime = \bigotimes_{i=1}^{2^{j+1}} \alice_i$, $\charlieprime = \bigotimes_{i'=2^{j+1}+1}^{2^{j+2}} \alice_{i'}$, and $\rho$ is split accordingly between $\bobprime, \charlieprime$. Now $\aliceprime$ is a distinguisher for $(\tvdist{y_0}{y_1}, \distr^{2^r t}, 2^r n)$-SSI, from which the claim follows.
    \end{proof}
    \begin{claim}   \label{clm:tv_12}
        $\tvdist{y_1}{y_2} \le \eps_j$.
    \end{claim}
    \begin{proof}
        Consider a $2^{j+1}$-party non-local adversary $\aliceprime=(\alice_1,\dots,\alice_{2^{j+1}-1},\bob, \rho)$, where $\bob$ is defined as follows: $\bob$ holds the last $2^{j+1} + 1$ registers of $\rho$, and on input $\ket{\phi}$ he simulates $\alice_{2^{j+1}+1},\dots, \alice_{2^{j+2}}$ by sampling their inputs $\ket{\psi'}^{\ot 2^r t} \from \distr^{2^r t}$. He then runs $\alice_{2^{j+1}}$ on input $\ket{\phi}$ and outputs the answer. Now $\aliceprime$ is a distinguisher for $(\tvdist{y_1}{y_2}, \distr^{2^{r-j-1} t}, 2^{r-j-1} n)$-SSI. Since $(\eps_j, \distr^{2^{r-j}t}, 2^{r-j} n)$-SSI implies $(\eps_j, \distr^{2^{r-j-1}t}, 2^{r-j-1} n)$-SSI, the claim follows.
    \end{proof}
    \begin{claim}   \label{clm:tv_23}
        $\tvdist{y_2}{y_3} \le \eps_j$.
    \end{claim}
    \begin{proof}
        Follows similarly as \Cref{clm:tv_12} by defining a $2^{j+1}$-party non-local adversary $\aliceprime=(\bob, \alice_{2^{j+1}+2},\dots,\alice_{2^{j+2}}, \rho)$ where $\bob$ simulates $\alice_1, \dots \alice_{2^{j+1}}$.
    \end{proof}
    Combining Claims \ref{clm:tv_01}, \ref{clm:tv_12}, \ref{clm:tv_23} and using the triangle inequality, we get $\tvdist{y_0}{y_3} \le \eps + 2\eps_j = (2^{j+2}-1)\eps = \eps_{j+1}$, so that $(\eps_{j+1}, \distr^{2^{r-j-1}t}, 2^{r-j-1}n)$-SSI holds against $2^{j+2}$ parties as desired. Therefore, by setting $j=r$ we conclude that $((2^{r+1}-1)\eps, \distr^t, n)$-SSI holds against $2^{r+1}\ge q$ parties, hence it also holds against $q$ parties. Since $(2^{r+1}-1)\eps < 2q\eps$, the lemma follows.
\end{proof}

Combining \Cref{lem:ssi_union_bound,lem:ssi_many_party} yields the following corollary: \begin{corollary}   \label{cor:ssi_many_party}
    Suppose that $(\eps, \distr^{qt})$-SSI holds, then $(2^{2qn}q\eps, \distr^t, n)$-SSI holds against $q$ parties.
\end{corollary}

\subsection{Distinguishing Bell-States}

In this section, we use the Bell basis over a 2-qubit system as a warm-up example to demonstrate some facts about simultaneous state indistinguishability. By the Bell basis we mean \begin{align*}
    \bracC { \frac{1}{\sqrt{2}}\brac{ \ket{00} + \ket{11} }, \frac{1}{\sqrt{2}}\brac{ \ket{00} - \ket{11} }, \frac{1}{\sqrt{2}}\brac{ \ket{\plus\plus} + \ket{\minus \minus} }, \frac{1}{\sqrt{2}}\brac{ \ket{\plus\plus} - \ket{\minus \minus} }  }.
\end{align*}
Note that the Bell basis is symmetric for the purposes of this section, meaning any fact we show about a pair of Bell states will also apply to any other pair of Bell states. We start off by demonstrating that the non-local distance $\nldist{\cdot}{\cdot}$ can be strictly less than LOCC distance. Recall that two orthogonal Bell states can be perfectly distinguished using LOCC measurements. 
\begin{lemma}[Comparison to LOCC]
    Let $\ket{\Phi^+}, \ket{\Phi^-}$ be orthogonal Bell states. Then, $\nldist{\ketbraX{\Phi^+}} {\ketbraX{\Phi^-}} = 1/2$.
\end{lemma}

\begin{proof}
     It suffices to consider projective $\bob, \charlie$. After that, the only non-trivial case is when they both are rank-1 projections, in which case it is easy to show
    \begin{align*}
     \tr \brac{ (\bob \otimes \charlie) \brac{ \ketbraX{\Phi^+} - \ketbraX{\Phi^-}}} \le \tr \brac{ (\bob \otimes \charlie)
\ketbraX{\Phi^+} } \le 1/2.
\end{align*}
\end{proof}
Next, we show that a non-local distinguisher with an arbitrary number of EPR pairs is no more powerful than unentangled distinguishers for the same problem.

\begin{lemma}[Entanglement has no effect]
    Let $\ket{\Phi^+}, \ket{\Phi^-}$ be orthogonal Bell states and let $\ket{\epr_n}$ denote $n$ EPR pairs. Then, $\nldist{\ketbraX{\Phi^+} \otimes \ketbraX{\epr_n}} {\ketbraX{\Phi^-}\otimes \ketbraX{\epr_n}} = 1/2.$
\end{lemma}
\begin{proof}
    Let $\ket{\Phi}$ be a Bell state. Then, the Schmidt decomposition of $\ket{\Phi} \otimes \ket{\epr_n}$ is given by \begin{align*}
        \ket{\Phi} \otimes \ket{\epr_n} = \frac{1}{\sqrt{2^{n+1}}} \sum_{i=1}^{2^{n+1}}{\ket{u_i}\ket{v_i}}.
    \end{align*}
    As before, we can assume the distinguisher $\bob, \charlie$ is projective. Let $r_B := \rk \bob, r_C := \rk \charlie$, $r := \min(r_B, r_C)$. We will show that \begin{align*}
        \left\| (\bob \otimes \charlie) \brac{ \ket{\Phi} \otimes \ket{\epr_n} } \right\|_2^2 \in \bracS{ \max \brac{ 0, \frac{r}{2^{n}} - 1 } , \frac{r}{2^{n+1}}},
    \end{align*}
    which implies that \begin{align*}
        \abs{ \left\| (\bob \otimes \charlie) \brac{ \ket{\Phi^+} \otimes \ket{\epr_n} } \right\|_2^2 - \left\| (\bob \otimes \charlie) \brac{ \ket{\Phi^-} \otimes \ket{\epr_n} } \right\|_2^2 } &\le \frac{r}{2^{n+1}} - \max \brac{ 0, \frac{r}{2^{n}} - 1 } \le \frac{1}{2}.
    \end{align*}
    We show the upper/lower bounds separately.
    \paragraph{Upper bound.} Without loss of generality assume $r_B = r$. Then,
    \begin{align*}
        \left\| (\bob \otimes \charlie) \brac{ \ket{\Phi} \otimes \ket{\epr_n} } \right\|_2^2 &\le \left\| (\bob \otimes \id) \brac{ \ket{\Phi} \otimes \ket{\epr_n} } \right\|_2^2 \\
        &= \tr \brac{ \bob \brac{ \tr_C \brac{ \ketbraX{\Phi} \otimes \ketbraX{\epr_n}} } } \\
        &= \tr \brac{ \bob \brac{\id/2^{n+1}} } \\
        &= \frac{r}{2^{n+1}}.
    \end{align*}
    \paragraph{Lower bound.} Since $0 \le \bob, \charlie \le \id$, we have \begin{align*}
        (\id - \bob) \otimes (\id - \charlie) \ge 0 \implies \bob \otimes \charlie \ge \bob \otimes \id + \id \otimes \charlie - \id \otimes \id.
    \end{align*}
    Thus, \begin{align*}
        &\quad \left\| (\bob \otimes \charlie) \brac{ \ket{\Phi} \otimes \ket{\epr_n} } \right\|_2^2 = \tr \brac{ (\bob \otimes \charlie) \brac{ \ketbraX{\Phi} \otimes \ketbraX{\epr_n}}} \\
        &\ge \tr \brac{ (\bob \otimes \id) \brac{ \ketbraX{\Phi} \otimes \ketbraX{\epr_n}}} + \tr \brac{ (\id \otimes \charlie) \brac{ \ketbraX{\Phi} \otimes \ketbraX{\epr_n}}} - 1 \\
        &= \frac{r_B}{2^{n+1}} + \frac{r_C}{2^{n+1}} - 1 \ge \frac{r}{2^n} - 1.
        \end{align*}
\end{proof}

\subsection{Impossibility Results about SSI}
\noindent There are many distributions that are simultaneously distinguishable. We give some examples below by identifying $\distr$ for which $(\eps,\distr)$-SSI does not hold for small $\eps$.

\begin{claim}
Suppose $\distr$ is a (classical) distribution on $\{0,1\}^n$ such that $\prob[y_1 \neq y_2: y_1,y_2 \from \distr] \geq c$, for some $c=c(n)$. Then, $(\eps,\distr)$-simultaneous state indistinguishability does not hold for any $\eps < c/2n$.
\end{claim}
\begin{proof} We will denote by $y^{i}$ the $i$-th bit of a string $y \in \bit^n$.
We construct $\adversary=(\bob,\charlie,\rho)$ as follows: 
\begin{itemize}
    \item $\rho = \brac{\frac{1}{n} \sum_{i \in [n]} \ketbra{i}{i}_\reg{B} \otimes \ketbra{i}{i}_\reg{C} } \otimes \brac{ \frac{1}{2} \sum_{b \in \bit} \ketbraX{b}_\reg{B'} \otimes \ketbraX{b}_\reg{C'}} $,
    \item $\bob$ upon input $y_{\bob}$, measures $\reg B$ to obtain $i \in [n]$ and measures $\reg{B'}$ to obtain $b$. $\bob$ outputs $y_{\bob}^{i} \oplus b$.
    \item $\charlie$ upon input $y_{\charlie}$, measures $\reg C$ to obtain $i \in [n]$ and measures $\reg{C'}$ to obtain $b$. $\charlie$ outputs $y_{\charlie}^{i} \oplus b$.
\end{itemize}
We will look at two cases. In the first case, $y \leftarrow \distr$ and both $\bob$ and $\charlie$ are given $y$. In this case, $(\bob, \charlie)$ will output $(0,1)$ with probability $0$.
\par In the second case, the probability that $\bob$ and $\charlie$ output $(0,1)$ is lower bounded by: 
\begin{align*}
 & \prob[(0,1) \leftarrow \adversary(y_{\bob},y_{\charlie})\ |\ y_\bob \ne y_\charlie ] \cdot \prob[y_\bob \ne y_\charlie] 
\ge \frac{1}{n} \cdot \frac{1}{2} \cdot c = \frac{c}{2n}
\end{align*}
since there is at least one pair of $(i,b)$ that results in $\alice$ outputting $(0,1)$ when $y_\bob \ne y_\charlie$.
\end{proof}
\noindent We remark that one can improve the bound $c/2n$ by having $\bob$ and $\charlie$ apply an error-correcting code on $y_\bob$ and $y_\charlie$ before measuring, respectively.

\begin{claim} \label{clm:cliffordattack}
Let $m,n \in \N$ and $m \ge n$. Let $C \in \linop(\C^{2^n}, \C^{2^m})$ be a Clifford circuit. That is, $C$ appends $m-n$ qubits initialized to zeros to its input and applies a sequence of Clifford gates. Define $\distr_1$ as follows: sample $r \xleftarrow{\$} \{0,1\}^{n-1}$ and then output $\ket{\psi_r}_\reg{B C} = C \ket{r||0}$. Define $\distr_2$ as follows: sample $r \xleftarrow{\$} \{0,1\}^{n-1}$ and then output $\ket{\psi_r}_\reg{B C} = C \ket{r||1}$. Then, the distributions $\distr_1$ and $\distr_2$ are $\eps$-simultaneous state distinguishable for all $\eps < \frac{1}{2}$.
\end{claim}
\begin{proof}
    We will describe a $1/2$-simultaneous distinguisher $(\bob, \charlie, \rho)$ for $\distr_1$ and $\distr_2$. $$\rho = \ketbraX{\epr}^{\otimes n} \otimes \brac{\frac{1}{2}(\ketbraX{0} \otimes \ketbraX{0} + \ketbraX{1} \otimes \ketbraX{1}}$$ is $n$ copies of a maximally entangled Bell state, which is given by $\ket{\epr} = 1/\sqrt{2} \brac{\ket{00} + \ket{11}}$, together with a shared random bit $b_S$. Recall that using one copy of the state $\epr$ and two bits of classical communication, $\bob$ can teleport a one-qubit state to $\charlie$. Even though communication is forbidden between $\bob$ and $\charlie$ in our setting, they can still perform the rest of the teleportation operation. In more detail, let $\ket{\psi_r}_\reg{B}$ denote the $m_\bob$-qubit state which is $\bob$'s half of $\ket{\psi_r}_\reg{B C}$, and let $\reg{B_i}$ be the register containing its $i$-th qubit. Let $\reg{B_i^{\epr}}$ be the register containing $\bob$'s half of the $i$-th Bell state. For $i \in [m_\bob]$, $\bob$ will do the following: \begin{itemize}
        \item Apply a $\cnot$ gate to $\reg{B_i}$ controlled on $\reg{B_i^{\epr}}$. Measure $\reg{B_i^{\epr}}$ in the computational basis and record the outcome as $a_i \in \bit$.
        \item Apply a Hadamard gate to $\reg{B_i}$ and then measure it in the computational basis. Record the outcome as $b_i \in \bit$.
    \end{itemize}

    Let $a = a_1\dots a_{m_\bob} \in \bit^{m_\bob}$ and $b = b_1\dots b_{m_\bob} \in \bit^{m_\bob}$. Observe that $\charlie$'s half the Bell states are now collapsed to the state $X^aZ^b \ket{\psi_r}_\reg{B}$. In other words, $\bob$ has teleported the state in a quantum one-time-padded form. The overall state of $\charlie$ is given by $\brac{X^a Z^b \otimes I_\bfC}\ket{\psi_r}_\reg{B C} = \brac{X^{\overline{a}} Z^{\overline{b}}}C\ket{r}\ket{\widetilde{b}}$ for some $\widetilde{b} \in \bit$, where $\overline{a} := a \| 0^{m-m_\bob}$ and $\overline{b} := b \| 0^{m-m_\bob}$. Since $C$ is a fixed Clifford circuit, there exist functions $f,g,h$ such that $C^\dagger \brac{X^{\overline{a}} Z^{\overline{b}}}C = e^{i\pi f(a,b)/2} X^{g(a,b)} Z^{h(a,b)}$. Hence, by applying $C^\dagger$, $\charlie$ obtains the state $e^{i\pi f(a,b)/2} X^{g(a,b)} Z^{h(a,b)} \ket{r} \ket{\widetilde{b}} = \ket{\Phi} \otimes \ket{\widetilde{b} \oplus g_m(a,b)}$, where $g_m$ is the $m$-th bit of $g$. Finally, \begin{itemize}
        \item $\bob$ computes $g_m(a,b)$ and outputs $g_m(a,b) \oplus b_S$.
        \item $\charlie$ measures the $m$-th qubit in the computational basis to obtain $\widetilde{b} \oplus g_m(a,b)$, and outputs $\widetilde{b} \oplus g_m(a,b) \oplus b_S$.
    \end{itemize}
    If $\widetilde{b} = 0$, then $(\bob, \charlie)$ output $(0,0)$ with probability $1/2$, whereas they never output $(0,0)$ if $\widetilde{b} = 1$. Thus, $(\bob, \charlie, \rho)$ is a $1/2$-simultaneous state distinguisher as desired.
\end{proof}

\renewcommand{\Haar}{\haar_{2^n}}

\section{Simultaneous Haar Indistinguishability}
\label{sec:haar_indist}

\noindent We show that SSI can be instantiated using Haar random states.

\subsection{Single-Copy Case}

\begin{theorem}[Simultaneous Haar Indistinguishability (SHI)] \label{conj:shi}
$(\eps,\haar_{2^n})$-simultaneous state indistinguishability holds for $\eps = O(1/2^{n/2})$.
\end{theorem}

\begin{proof}
    Let $d=2^n$ be the dimension of the Haar states. Consider a non-local adversary $\alice = (\bob, \charlie, \psi)$, where
    $\bob$ and $\charlie$ share an entangled state $\ket{\psi}_{\reg{B_1 C_1}}$. Bob receives a state $\ket{b}_{\reg{B_2}}$ and Charlie receives a state $\ket{c}_{\reg{C_2}}$ from the referee. By \Cref{lem:ssi_proj}, we can assume that Bob applies a projective measurement $M_{\reg{B_1 B_2}}$ and Charlie applies a projective measurement $N_{\reg{C_1 C_2}}$. Let $p(\ket{b},\ket{c})$ denote the probability that both Bob's and Charlie's measurements accept:
\[
    p(\ket{b},\ket{c}) := \Tr( (M \ot N)( \psi \ot b \ot c))~.
\]
Therefore, the advantage we need to bound is given by
\[
    \max_{M,N} \abs{ \E_{\ket{\theta} \from \Haar} p(\ket{\theta},\ket{\theta}) - \E_{\ket{b},\ket{c} \from \Haar} p(\ket{b},\ket{c}) }
\]
where in the second expectation, the states $\ket{b},\ket{c}$ are sampled independently from the Haar measure. We will show that this is bounded by $O(1/\sqrt{d})$.

Note that by \cref{eq:sym_sub} we have
\[
    \E_{\ket{\theta} \from \Haar} \theta \ot \theta = \frac{\SymmSub(\reg{B_2 C_2)}}{\dim(\SymmSubspace(\reg{B_2 C_2}))},
\]
and also
\[
    \E_{\ket{b},\ket{c} \from \Haar} b \ot c = \frac{\id_{\reg{B_2}}}{\dim(\reg{B_2})} \ot \frac{\id_{\reg{C_2}}}{\dim(\reg{C_2})}~.
\]
\noindent
By assumption, $d = \dim(\reg{B_2}) = \dim(\reg{C_2})$. Then $\dim(\SymmSubspace(\reg{B_2 C_2})) = d(d + 1)/2$. Furthermore, by \Cref{lem:sum_perm} we have 
\[
\SymmSub(\reg{B_2 C_2)} = \frac{1}{2}( \id_{\reg{B_2}} \ot \id_{\reg{C_2}} + F_{\reg{B_2 C_2}})
\]
where $F_{\reg{B_2 C_2}}$ denotes the swap operator on registers $\reg{B_2}, \reg{C_2}$. Thus we can rewrite our advantage as follows:
\begin{align*}
    &\max_{M,N} \abs{ \Tr \Big( (M \ot N) \, \left \{ \psi_{\reg{B_1 C_1}} \ot \Big ( \frac{ \id_{\reg{B_2}} \ot \id_{\reg{C_2}} + F_{\reg{B_2 C_2}}}{d (d + 1)} -  \frac{\id_{\reg{B_2}} \ot \id_{\reg{C_2}}}{d^2 } \Big) \right \} \Big)} \\
    &\leq \max_{M,N} \abs{ \Tr \Big( (M \ot N) \, \left \{ \psi_{\reg{B_1 C_1}} \ot \Big ( \frac{ \id_{\reg{B_2}} \ot \id_{\reg{C_2}} + F_{\reg{B_2 C_2}}}{d^2} -  \frac{\id_{\reg{B_2}} \ot \id_{\reg{C_2}}}{d^2} \Big) \right \} \Big)} + O\brac{\frac{1}{d}} \\
    &=\max_{M,N} \frac{1}{d^2} \Tr \Big( (M \ot N) \, \left \{ \psi_{\reg{B_1 C_1}} \ot  F_{\reg{B_2 C_2}}  \right \} \Big) + O\brac{\frac{1}{d}}
\end{align*}
Using the Schmidt decomposition, we can write
\[
    \ket{\psi}_{\reg{B_1 C_1}} = (\sigma^{1/2}_{\reg{B_1}} \ot V_{\reg{C_1}}) \ket{\Omega}_{\reg{B_1 C_1}}
\]
where $\sigma_{\reg{B_1}}$ is a density matrix, $V_{\reg{C_1}}$ is a unitary operator and $\ket{\Omega}$ is the \emph{unnormalized} maximally entangled state on registers $\reg{B_1 C_1}$. Without loss of generality we can assume that $V = \id$, because we can always conjugate $N$ by this unitary. Therefore for fixed $M, N$ we can rewrite the above trace as
\[
    \frac{1}{d^2} \Tr \Big( (\tilde{M} \ot N) \, \left \{ \Omega_{\reg{B_1 C_1}} \ot  F_{\reg{B_2 C_2}}  \right \} \Big)
\]
where 
\[
    \tilde{M}_{\reg{B_1 B_2}} := (\sigma^{1/2}_{\reg{B_1}} \ot \id_{\reg{B_2}}) M (\sigma^{1/2}_{\reg{B_1}} \ot \id_{\reg{B_2}})~.
\]
One can verify using tensor network diagrams 
(see \Cref{figure:tensor} for a graphical proof) that 
\begin{equation}
    \label{eq:tensor}
    \Tr \Big( (\tilde{M} \ot N) \, \left \{ \Omega_{\reg{B_1 C_1}} \ot  F_{\reg{B_2 C_2}}  \right \} \Big) = \Tr( \tilde{M} \cdot N^{\top_{\reg{C_1}}})
\end{equation}
where $N^{\top_{\reg{C_1}}}$ denotes the \emph{partial transpose} of the operator $N$ on register $\reg{C_1}$ and the matrix multiplication ($\cdot$) connects registers $\reg{B_1}$ with $\reg{C_1}$ and $\reg{B_2}$ with $\reg{C_2}$. 

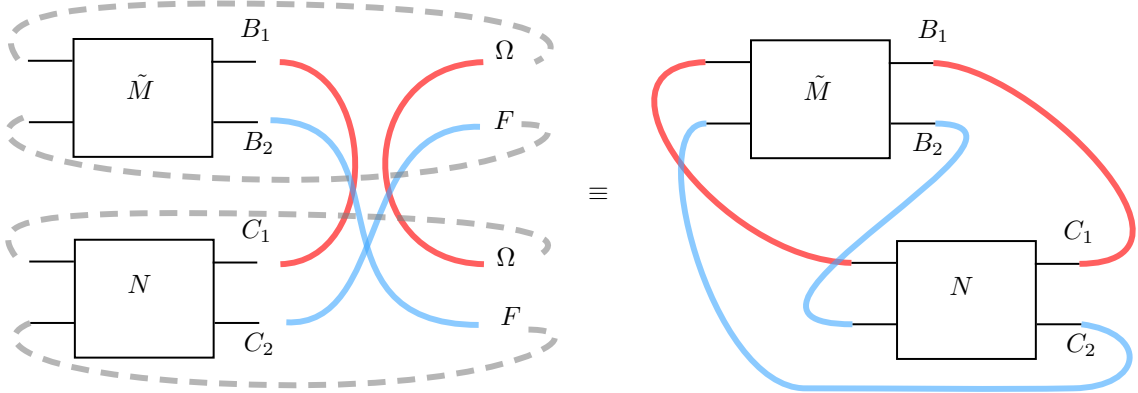
\begin{figure}[!htb]
\begin{center}
\tikzset{every picture/.style={line width=0.75pt}} 

\begin{tikzpicture}[x=0.75pt,y=0.75pt,yscale=-1,xscale=1]

\draw   (43,38) -- (113,38) -- (113,97.5) -- (43,97.5) -- cycle ;
\draw    (112.5,49.5) -- (135,49.5) ;
\draw    (113.5,80) -- (136,80) ;
\draw    (20,49) -- (42.5,49) ;
\draw    (20.5,80) -- (43,80) ;
\draw   (43.67,139.33) -- (113.67,139.33) -- (113.67,198.83) -- (43.67,198.83) -- cycle ;
\draw    (113.17,150.83) -- (135.67,150.83) ;
\draw    (114.17,181.33) -- (136.67,181.33) ;
\draw    (20.67,150.33) -- (43.17,150.33) ;
\draw    (21.17,181.33) -- (43.67,181.33) ;
\draw [color={rgb, 255:red, 255; green, 0; blue, 0 }  ,draw opacity=0.63 ][line width=2.25]    (147,49.67) .. controls (196.67,50.33) and (198.67,150.33) .. (147.33,151.67) ;
\draw [color={rgb, 255:red, 255; green, 0; blue, 0 }  ,draw opacity=0.63 ][line width=2.25]    (249.67,49.67) .. controls (184,51.67) and (183.33,150.33) .. (250,151.67) ;
\draw [color={rgb, 255:red, 65; green, 169; blue, 255 }  ,draw opacity=0.61 ][line width=2.25]    (142.33,79) .. controls (214,79.67) and (157.33,182.33) .. (247.33,182.33) ;
\draw [color={rgb, 255:red, 65; green, 169; blue, 255 }  ,draw opacity=0.61 ][line width=2.25]    (150.33,181) .. controls (202.67,183.67) and (185.33,80.33) .. (248,82.33) ;
\draw [color={rgb, 255:red, 127; green, 127; blue, 127 }  ,draw opacity=0.58 ][line width=2.25]  [dash pattern={on 6.75pt off 4.5pt}]  (15,48.33) .. controls (1,29) and (24.33,18.33) .. (145,20.33) .. controls (265.67,22.33) and (301.67,37.67) .. (275.67,49.67) ;
\draw [color={rgb, 255:red, 127; green, 127; blue, 127 }  ,draw opacity=0.58 ][line width=2.25]  [dash pattern={on 6.75pt off 4.5pt}]  (17.67,77.67) .. controls (-5.67,87) and (19.67,107) .. (140.33,109) .. controls (261,111) and (308.33,81) .. (268.33,81) ;
\draw [color={rgb, 255:red, 127; green, 127; blue, 127 }  ,draw opacity=0.58 ][line width=2.25]  [dash pattern={on 6.75pt off 4.5pt}]  (13.67,147.67) .. controls (-0.33,128.33) and (23,124.33) .. (143.67,126.33) .. controls (264.33,128.33) and (300.33,137) .. (274.33,149) ;
\draw [color={rgb, 255:red, 127; green, 127; blue, 127 }  ,draw opacity=0.58 ][line width=2.25]  [dash pattern={on 6.75pt off 4.5pt}]  (21.17,181.33) .. controls (-2.17,190.67) and (23.17,210.67) .. (143.83,212.67) .. controls (264.5,214.67) and (311.83,184.67) .. (271.83,184.67) ;
\draw   (384.67,38.67) -- (454.67,38.67) -- (454.67,98.17) -- (384.67,98.17) -- cycle ;
\draw    (454.17,50.17) -- (476.67,50.17) ;
\draw    (455.17,80.67) -- (477.67,80.67) ;
\draw    (361.67,49.67) -- (384.17,49.67) ;
\draw    (362.17,80.67) -- (384.67,80.67) ;
\draw   (458.33,140) -- (528.33,140) -- (528.33,199.5) -- (458.33,199.5) -- cycle ;
\draw    (527.83,151.5) -- (550.33,151.5) ;
\draw    (528.83,182) -- (551.33,182) ;
\draw    (435.33,151) -- (457.83,151) ;
\draw    (435.83,182) -- (458.33,182) ;
\draw [color={rgb, 255:red, 255; green, 0; blue, 0 }  ,draw opacity=0.63 ][line width=2.25]    (476.67,50.17) .. controls (526.33,50.83) and (624.67,151) .. (550.33,151.5) ;
\draw [color={rgb, 255:red, 255; green, 0; blue, 0 }  ,draw opacity=0.63 ][line width=2.25]    (361.67,49.67) .. controls (296,51.67) and (368.67,149.67) .. (435.33,151) ;
\draw [color={rgb, 255:red, 65; green, 169; blue, 255 }  ,draw opacity=0.61 ][line width=2.25]    (477.67,80.67) .. controls (549.33,81.33) and (345.83,182) .. (435.83,182) ;
\draw [color={rgb, 255:red, 65; green, 169; blue, 255 }  ,draw opacity=0.61 ][line width=2.25]    (551.33,182) .. controls (585.38,186.39) and (583.33,213.67) .. (547.33,214.33) .. controls (511.33,215) and (460,214.33) .. (412,214.33) .. controls (364,214.33) and (329.49,79.62) .. (362.17,80.67) ;

\draw (68,54.4) node [anchor=north west][inner sep=0.75pt]    {$\tilde{M}$};
\draw (125.5,26.4) node [anchor=north west][inner sep=0.75pt]    {$B_{1}$};
\draw (126.75,83.4) node [anchor=north west][inner sep=0.75pt]    {$B_{2}$};
\draw (68.67,155.73) node [anchor=north west][inner sep=0.75pt]    {$N$};
\draw (126.17,127.73) node [anchor=north west][inner sep=0.75pt]    {$C_{1}$};
\draw (127.42,184.73) node [anchor=north west][inner sep=0.75pt]    {$C_{2}$};
\draw (254.17,37.73) node [anchor=north west][inner sep=0.75pt]    {$\Omega $};
\draw (256.83,171.73) node [anchor=north west][inner sep=0.75pt]    {$F$};
\draw (253.5,73.07) node [anchor=north west][inner sep=0.75pt]    {$F$};
\draw (254.83,141.73) node [anchor=north west][inner sep=0.75pt]    {$\Omega $};
\draw (409.67,55.07) node [anchor=north west][inner sep=0.75pt]    {$\tilde{M}$};
\draw (467.17,27.07) node [anchor=north west][inner sep=0.75pt]    {$B_{1}$};
\draw (464.42,84.07) node [anchor=north west][inner sep=0.75pt]    {$B_{2}$};
\draw (483.33,156.4) node [anchor=north west][inner sep=0.75pt]    {$N$};
\draw (540.83,128.4) node [anchor=north west][inner sep=0.75pt]    {$C_{1}$};
\draw (542.08,185.4) node [anchor=north west][inner sep=0.75pt]    {$C_{2}$};
\draw (301,111.73) node [anchor=north west][inner sep=0.75pt]    {$\equiv $};

\end{tikzpicture}
\end{center} 
\caption{
A graphical proof of the equivalence of \Cref{eq:tensor}. On the left is a tensor network diagram depicting the trace inner product between $\tilde{M} \ot N$ and $\Omega \ot F$; the grey dashed lines denote the trace; the red wires denote the (unnormalized) maximally entangled state $\Omega$, and the blue wires denote the swap operation $F$. By following the wire connections we deduce the tensor network diagram on the right, which is the same as the trace inner product between $\tilde{M}$ and the partial transpose of $N$ on register $\reg{C}_1$. 
}
\label{figure:tensor}
\end{figure} 

For now let's suppose that $\sigma = \id_{\reg{B_1}}/\dim(\reg{B_1})$ (i.e., the state $\ket{\psi}$ is the \emph{normalized} maximally entangled state). Then $\tilde{M} = M/(\dim \reg{B_1})$ and we can thus upper-bound the trace by
\[
    \Tr( \tilde{M} \cdot N^{\top_{\reg{C_1}}}) \leq \| \tilde{M} \|_2 \cdot \| N^{\top_{\reg{C_1}}} \|_2 = \frac{1}{\dim(\reg{B_1})} \| M \|_2 \cdot \| N \|_2
\]
by the Cauchy-Schwarz inequality and by the fact that the partial transpose does not increase the Hilbert-Schmidt norm. But now observe that 
\[
    \| M \|_2 \cdot \| N\|_2 \leq \dim(\reg{B_1}) \cdot d~.
\]
Therefore the maximum advantage is bounded from above by $1/d + O(1/d) = O(1/d)$.

\par What if $\ket{\psi}$ is not maximally entangled? We can rewrite things as follows: using that $(N^\top)^{\top_{\reg{C_2}}} = N^{\top_{\reg{C_1}}}$, 
\begin{align*}
    \Tr( \tilde{M} \cdot N^{\top_{\reg{C_1}}}) &= \Tr( (\sigma^{1/2} \ot \id) M  (\sigma^{1/2} \ot \id)(N^\top)^{\top_{\reg{C_2}}})~.   
\end{align*}
Now we insert resolutions of the identity on the $\reg{B_2},\reg{C_2}$

registers:
\begin{align*}
 \Tr( \tilde{M} \cdot N^{\top_{\reg{C_1}}}) &= \sum_{i,j} \Tr(  (\sigma^{1/2} \ot \ketbra{i}{i}) M (\sigma^{1/2} \ot \ketbra{j}{j}) \Big({N}^\top  \Big)^{\top_{\reg{C_2}}}) \\
 &= \sum_{i,j} \Tr(  (\sigma^{1/2} \ot \ketbra{i}{i}) M (\id \otimes \ketbra{j}{j}) (\sigma^{1/2} \ot \ketbra{j}{j}) \Big({N}^\top  \Big)^{\top_{\reg{C_2}}} (\id \otimes \ketbra{i}{i}))
\end{align*}
Note that
\[
    (\sigma^{1/2} \otimes \ketbra{j}{j}) \Big({N}^\top  \Big)^{\top_{\reg{C_2}}} (\id \otimes \ketbra{i}{i}) = (\sigma^{1/2} \otimes \ketbra{j}{i}) {N}^\top  (\id \otimes \ketbra{j}{i})~.
\]
Therefore by Cauchy-Schwarz we have 
\begin{align}
   \Tr( \tilde{M} \cdot N^{\top_{\reg{C_1}}}) &\leq \sum_{i,j} \Big \| (\sigma^{1/2} \ot \ketbra{i}{i}) M (\id \ot \ketbra{j}{j}) \Big \|_2 \cdot \Big \| (\sigma^{1/2} \ot \ketbra{j}{i}) {N}^\top  (\id \ot \ketbra{j}{i})  \Big \|_2~. \label{eq:pt}
\end{align}
We bound the second factor:
\begin{align*}
    \Big \| (\sigma^{1/2} \ot \ketbra{j}{i}) {N}^\top  (\id \ot \ketbra{j}{i})  \Big \|_2^2 &= \Tr \Big( (\sigma \ot \ketbra{i}{i}) {N}^\top (\id \otimes \ketbra{j}{j}) \overline{N}  \Big) \leq 1
\end{align*}
where in the equality we used the fact that $N^\top$ is a projection and therefore ${N}^\top (\id \otimes \ketbra{j}{j}) \overline{N}$ has operator norm at most $1$. 

Continuing we have that
\begin{align*}
    \Tr( \tilde{M} \cdot N^{\top_{\reg{C_1}}}) &\leq \sum_{i,j} \Big \| (\sigma^{1/2} \ot \ketbra{i}{i}) M (\id \ot \ketbra{j}{j}) \Big \|_2 \\
    &\leq d \sqrt{ \sum_{i,j} \Big \| (\sigma^{1/2} \ot \ketbra{i}{i}) M (\id \ot \ketbra{j}{j}) \Big \|_2^2 } \\
    &= d \sqrt{ \sum_{i,j} \Tr \Big( (\sigma \ot \ketbra{i}{i}) M (\id \ot \ketbra{j}{j}) M \Big ) } \\
    &= d \sqrt{ \sum_{i} \Tr \Big( (\sigma \ot \ketbra{i}{i}) M \Big ) } \\
    &\leq d \, \sqrt{d}~.
\end{align*}
In the last inequality, we used the fact that $\Tr \Big( (\sigma \ot \ketbra{i}{i}) M \Big ) \leq 1$ for all $i$. 

Putting everything together, this yields that the maximum advantage satisfies the bound
\[
     \max_{M,N} \abs{ \E_{\ket{\theta} \from \Haar} p(\ket{\theta},\ket{\theta}) - \E_{\ket{b},\ket{c} \from \Haar} p(\ket{b},\ket{c})} \leq \frac{1}{\sqrt{d}} + O\brac{\frac{1}{d}} = O \Big (\frac{1}{\sqrt{d}} \Big)~.
\]

\end{proof}
\begin{remark} \label{rem:unbounded_ent}
    The entangled state held by $\alice$ above can have arbitrary dimension. In particular, the proof shows that an arbitrary number of EPR pairs has no asymptotic advantage in distinguishing identical vs. independent Haar states, for the advantage is $O(1/d)$ in both cases (see \Cref{sec:haar_lb} for the lower bound).
\end{remark}

\subsection{Generalization to Many Copies}
We generalize \Cref{conj:shi} to show indistinguishability when the non-local distinguisher is given many copies of the input state. The proof has the same overall structure, but requires some additional ideas. We start by listing a useful lemma.

\begin{lemma}\label{lem:perm}
    Let $S_{2t}^{(s)} \subset S_{2t}$ be the set of permutations $\sigma$ over $[2t]$ such that $$ \abs{ \bracC{i \in [t] \; : \; \sigma(i) \notin [t]} } = s. $$ Define $P_\sigma$ as the permutation operator over registers $\reg{X_1 \dots X_t Y_1 \dots Y_t}$, where we match $[t]$ with $\reg{X_1 \dots X_t}$ and $[2t]\setminus [t]$ with $\reg{Y_1 \dots Y_t}$. Then, \begin{align*}
        \sum_{\sigma \in S_{2t}^{(s)}} P_\sigma = (t!)^2 \binom{t}{s}^2 \brac{\SymmSub(\reg{X_1\dots X_t}) \ot \SymmSub(\reg{Y_1 \dots Y_t})} P_{\sigma_s} \brac{\SymmSub(\reg{X_1\dots X_t}) \ot \SymmSub(\reg{Y_1 \dots Y_t})},
    \end{align*}
    where $\sigma_s \in S_{2t}^{(s)}$ is the permutation that swaps $i$ with $t+i$ for $i \in [s]$.
\end{lemma}
\begin{proof}
    Every $\sigma \in S_{2t}^{(s)}$ can be written as $\brac{P_{\sigma_X} \ot P_{\sigma_Y}} P_{\sigma_s} \brac{P_{\sigma_X'}\ot P_{\sigma_Y'}} $ for some $\sigma_X,\sigma_Y,\sigma_X',\sigma_Y' \in S_t$. Moreover, picking $\sigma_X,\sigma_Y,\sigma_X',\sigma_Y'$ uniformly is equivalent to picking $\sigma$ uniformly. Observe by a counting argument that $$\abs{S_{2t}^{(s)}} = (t!)^2 \binom{t}{s}^2.$$
    Thus, by \Cref{lem:sum_perm} we have \begin{align*}
        \frac{1}{(t!)^2 \binom{t}{s}^2} \sum_{\sigma \in S_{2t}^{(s)}} P_\sigma &= \frac{1}{(t!)^4} \sum_{\sigma_X,\sigma_Y,\sigma_X',\sigma_Y' \in S_t} \brac{P_{\sigma_X} \ot P_{\sigma_Y}} P_{\sigma_s} \brac{P_{\sigma_X'}\ot P_{\sigma_Y'}} \\
        &= \brac{\SymmSub(\reg{X_1\dots X_t}) \ot \SymmSub(\reg{Y_1 \dots Y_t})} P_{\sigma_s} \brac{\SymmSub(\reg{X_1\dots X_t}) \ot \SymmSub(\reg{Y_1 \dots Y_t})}
    \end{align*}
\end{proof}

\begin{theorem}[$t$-Copy Simultaneous Haar Indistinguishability (SHI)] \label{thm:shi_many_copy} Let $\eps = o(1)$ and $t = \eps 2^{n/4}$. Let $\haar^t_{2^n}$ be the distribution defined by sampling $t$ copies of a state from $\haar_{2^n}$. Then, $(O(\eps^2),\haar^t_{2^n})$-simultaneous state indistinguishability holds. 
\end{theorem}

\begin{proof}
    Let $d=2^n$ be the dimension of the Haar states. Consider a non-local adversary $\alice = (\bob, \charlie, \psi)$, where
    $\bob$ and $\charlie$ share an entangled state $\ket{\psi}_{\reg{B_1 C_1}}$. Bob receives a state $\ket{b}_{\reg{B_2}}$ and Charlie receives a state $\ket{c}_{\reg{C_2}}$ from the referee. By \Cref{lem:ssi_proj}, we can assume that Bob applies a projective measurement $M_{\reg{B_1 B_2}}$ and Charlie applies a projective measurement $N_{\reg{C_1 C_2}}$. Let $p(\ket{b},\ket{c})$ denote the probability that both Bob's and Charlie's measurements accept:
\[
    p(\ket{b},\ket{c}) := \Tr( (M \ot N)( \psi \ot b \ot c))~.
\]
Therefore, the advantage we need to bound is given by
\[
    \max_{M,N} \abs{ \E_{\ket{\theta} \from \Haar} p(\ket{\theta}^{\ot t},\ket{\theta}^{\ot t}) - \E_{\ket{\theta},\ket{\theta'} \from \Haar} p(\ket{\theta}^{\ot t},\ket{\theta'}^{\ot t}) }
\]
where in the second expectation, the states $\ket{\theta},\ket{\theta'}$ are sampled independently from the Haar measure. We will show that this is bounded by $d^{-\Omega(1)}$. For the $t$ copies of the Haar state, we can decompose the registers $\reg{B_2,C_2}$ as $\reg{B_2} = \reg{B_2^1\dots B_2^t}$ and $\reg{C_2} = \reg{C_2^1\dots C_2^t}$.

Note that by \cref{eq:sym_sub} we have
\[
    \E_{\ket{\theta} \from \Haar} \theta^{\ot t} \ot \theta^{\ot t} = \frac{\SymmSub(\reg{B_2 C_2)}}{\dim(\SymmSubspace(\reg{B_2 C_2}))}
\]
and
\[
    \E_{\ket{b},\ket{c} \from \Haar} b^{\ot t} \ot c^{\ot t} = \frac{\SymmSub({\reg{B_2}})}{\dim(\SymmSubspace(\reg{B_2}))} \ot \frac{\SymmSub({\reg{C_2}})}{\dim(\SymmSubspace(\reg{C_2}))}~.
\]

By assumption, $d = \dim(\reg{B_2^j}) = \dim(\reg{C_2^j})$ for $j \in [t]$. Then $\dim(\SymmSubspace(\reg{B_2 C_2})) = \binom{d+2t-1}{2t}$ and $\dim(\SymmSubspace(\reg{B_2})) = \dim(\SymmSubspace(\reg{C_2})) = \binom{d+t-1}{t}$. Furthermore,
\[
\SymmSub(\reg{B_2 C_2)} = \frac{1}{(2t)!} \sum_{\sigma \in S_{2t}} P_\sigma \quad \text{and} \quad \SymmSub(\reg{B_2}) = \frac{1}{(t)!} \sum_{\sigma \in S_{t}} P_{\sigma_B}, \quad \SymmSub(\reg{C_2}) = \frac{1}{(t)!} \sum_{\sigma \in S_{t}} P_{\sigma_C},
\]
where $P_\sigma$ is a permutation operator over $(2t)$ registers $B_2C_2$, and $P_{\sigma_B},P_{\sigma_C}$ are permutation operators over $(t)$ registers $B_2,C_2$, respectively.
Thus, using the collision bound we can rewrite our advantage as the following:
\begin{align*}
    &\max_{M,N} \abs{ \Tr \brac{ (M \ot N) \, \left \{ \psi_{\reg{B_1 C_1}} \ot \brac{ \frac{ \sum_{\sigma \in S_{2t}} P_\sigma }{d (d + 1)\dots(d+2t-1)} -  \frac{\sum_{\sigma_B,\sigma_C \in S_t} P_{\sigma_B}\ot P_{\sigma_C}}{d^2(d+1)^2\dots(d+t-1)^2 } } \right \} }} \\
    &\leq \max_{M,N} \abs{ \Tr \brac{ (M \ot N) \, \left \{ \psi_{\reg{B_1 C_1}} \ot \brac{ \frac{ \sum_{\sigma \in S_{2t}} P_\sigma }{d^{2t}} -  \frac{\sum_{\sigma_B,\sigma_C \in S_t} P_{\sigma_B}\ot P_{\sigma_C}}{d^{2t} } } \right \} }} + O\brac{ \frac{t^2}{d} } \\
    &=\max_{M,N} \frac{1}{d^{2t}} \Tr \Big( (M \ot N) \, \Big \{ \psi_\reg{B_1 C_1} \ot \sum_{\sigma \in S_{2t}^*} P_\sigma  \Big \} \Big) + O\brac{\frac{t^2}{d}},
\end{align*}
where $S_{2t}^*$ denotes the set of permutations $\sigma \in S_{2t}$ which \emph{cannot} be expressed as a product of permutations $\sigma_B,\sigma_C \in S_t$.
Using the Schmidt decomposition, we can write
\[
    \ket{\psi}_{\reg{B_1 C_1}} = (\sigma^{1/2}_{\reg{B_1}} \ot V_{\reg{C_1}}) \ket{\Omega}_{\reg{B_1 C_1}}
\]
where $\sigma_{\reg{B_1}}$ is a density matrix, $V_{\reg{C_1}}$ is a unitary operator and $\ket{\Omega}$ is the \emph{unnormalized} maximally entangled state on registers $\reg{B_1 C_1}$. Without loss of generality we can assume that $V = \id$, because we can always conjugate $N$ by this unitary. Therefore for fixed $M, N$ we can rewrite the above trace as
\[
    \frac{1}{d^{2t}} \Tr \Big( (\tilde{M} \ot N) \, \Big \{ \Omega_\reg{B_1 C_1} \ot \sum_{\sigma \in S_{2t}^*} P_\sigma  \Big \} \Big) 
\]
where 
\[
    \tilde{M}_{\reg{B_1 B_2}} := (\sigma^{1/2}_{\reg{B_1}} \ot \id_{\reg{B_2}}) M (\sigma^{1/2}_{\reg{B_1}} \ot \id_{\reg{B_2}})~.
\]

\par As in \Cref{lem:perm}, let $S_{2t}^{(s)}$ be the set of permutations over $2t$ registers ($\reg{B_2 C_2}$) which map $s$ registers from $\reg{B_2}$ to $\reg{C_2}$. Note that \begin{align*}
    |S_{2t}^{(s)}| = (t!)^2 \binom{t}{s}^2 \quad \text{and} \quad S_{2t}^* = \bigcup_{s=1}^t S_{2t}^{(s)}.
\end{align*}
Therefore for fixed $M, N$ we can rewrite the above trace as
\begin{align}
    &\frac{1}{d^{2t}} \Tr \Big( (\tilde{M} \ot N) \, \Big \{ \Omega_\reg{B_1 C_1} \ot \sum_{\sigma \in S_{2t}^*} P_\sigma  \Big \} \Big) = \frac{1}{d^{2t}} \sum_{s=1}^t \Tr \Big( (\tilde{M} \ot N) \, \Big \{ \Omega_\reg{B_1 C_1} \ot \sum_{\sigma \in S_{2t}^{(s)}} P_\sigma  \Big \} \Big) \nonumber \\
    &= \frac{(t!)^2}{d^{2t}} \sum_{s=1}^t \binom{t}{s}^2 \Tr \Big( (\tilde{M} \ot N) \, \Big \{ \Omega_\reg{B_1 C_1} \ot \brac{\SymmSub(\reg{B_2}) \ot \SymmSub(\reg{C_2})} P_{\sigma_s} \brac{\SymmSub(\reg{B_2}) \ot \SymmSub(\reg{C_2})} \Big \} \Big) \label{eq:x}
\end{align}
using \Cref{lem:perm}, where $P_{\sigma_s}$ swaps $\reg{B_2^{1}\dots B_2^s}$ with $\reg{C_2^{1}\dots C_2^s}$.
Define 
\[
\hat{M} := \brac{\id_{\reg{B_1}}\ot \SymmSub(\reg{B_2}) }\tilde{M} \brac{\id_{\reg{B_1}}\ot \SymmSub(\reg{B_2}) }, \quad \hat{N} := \brac{\id_{\reg{C_1}}\ot \SymmSub(\reg{C_2}) } N \brac{\id_{\reg{C_1}}\ot \SymmSub(\reg{C_2}) },
\]
then using cyclicity of trace we can rewrite \cref{eq:x} as \begin{align*}
    \frac{(t!)^2}{d^{2t}} \sum_{s=1}^t \binom{t}{s}^2 \Tr \Big( (\hat{M} \ot \hat{N}) \, \Big \{ \Omega_\reg{B_1 C_1} \ot P_{\sigma_s} \Big \} \Big).
\end{align*}
One can verify using tensor network diagrams\footnote{This follows by an easy generalization of \Cref{figure:tensor}.} that 
\begin{align*}
    \Tr \Big( (\hat{M} \ot \hat{N}) \, \Big \{ \Omega_\reg{B_1 C_1} \ot P_{\sigma_s} \Big \} \Big) = \Tr\brac{ \tr_{\reg{B_2^{[t]\setminus [s]}}}\hat{M} \cdot \tr_{\reg{C_2^{[t]\setminus [s]}}}(
 \hat{N})^{\top_{\reg{C_1}}}},
\end{align*}
where ${\top_{\reg{C_1}}}$ denotes the \emph{partial transpose} on register $\reg{C_1}$ and the matrix multiplication connects registers $\reg{B_1}$ with $\reg{C_1}$ and $\reg{B_2^{[s]}}$ with $\reg{C_2^{[s]}}$ according to $\sigma_s$. Suppose $\sigma_{\reg{B_1}} = \id_{\reg{B_1}}/\dim(\reg{B_1})$, i.e. $\psi$ is maximally entangled and $\tilde{M} = M/\dim(\reg{B_1})$. Note that \[
\hat{M} \le \frac{1}{\dim(\reg{B_1})} \brac{\id_{\reg{B_1}}\ot \SymmSub(\reg{B_2}) } \quad \text{and} \quad \hat{N} \le \id_{\reg{B_1}}\ot \SymmSub(\reg{C_2}).
\] Using Cauchy-Schwarz as well as \Cref{lem:monotone,lem:sym_part_tr}, we get \begin{align*}
    \abs{\Tr\brac{ \tr_{\reg{B_2^{[t]\setminus [s]}}}\hat{M} \cdot \tr_{\reg{C_2^{[t]\setminus [s]}}}(
 \hat{N})^{\top_{\reg{C_1}}}}} &\le \ellpnorm{\tr_{\reg{B_2^{[t]\setminus [s]}}}\hat{M}}{2} \ellpnorm{\tr_{\reg{C_2^{[t]\setminus [s]}}}(
 \hat{N})^{\top_{\reg{C_1}}}}{2} \\
 &= \ellpnorm{\tr_{\reg{B_2^{[t]\setminus [s]}}}\hat{M}}{2} \ellpnorm{\tr_{\reg{C_2^{[t]\setminus [s]}}}(
 \hat{N})}{2} \\
 &\le \frac{1}{\dim(\reg{B_1})}\ellpnorm{\id_{\reg{B_1}} \ot \tr_{\reg{B_2^{[t]\setminus [s]}}} \SymmSub(\reg{B_2})}{2}\ellpnorm{\id_{\reg{C_1}} \ot \tr_{\reg{C_2^{[t]\setminus [s]}}} \SymmSub(\reg{C_2})}{2} \\
 &\le \frac{\binom{d+t-1}{t}^2}{\dim(\reg{B_1})\binom{d+s-1}{s}^2} \ellpnorm{\id_{\reg{B_1}} \ot \SymmSub(\reg{B_2^{[s]}})}{2} \ellpnorm{\id_{\reg{C_1}} \ot \SymmSub(\reg{C_2^{[s]}})}{2} \\
 &= \frac{\binom{d+t-1}{t}^2}{\binom{d+s-1}{s}} \le \brac{1 + O\brac{\frac{t^2}{d}}} \frac{d^{2t-s}s!}{(t!)^2}
\end{align*}

Therefore, we can upper-bound \cref{eq:x} by \begin{align*}
    \sum_{s=1}^t \binom{t}{s}^2 s! d^{-s} \le \sum_{s=1}^t \binom{t}{s} \brac{\frac{t}{d}}^s = \brac{ 1 + \frac{t}{d} }^t - 1 \le e^{t^2/d} - 1 \le O\brac{t^2/d} = O(\eps^2/\sqrt{d}).
\end{align*}
\\
\par Now onto the case when $\sigma_\reg{B_1}$ is an arbitrary density matrix. Define \begin{align*}
    X := \tr_{\reg{B_2^{[t]\setminus [s]}}} \bracS{ \brac{\id_{\reg{B_1}}\ot \SymmSub(\reg{B_2}) }M \brac{\id_{\reg{B_1}}\ot \SymmSub(\reg{B_2}) }}, \quad Y := \tr_{\reg{C_2^{[t]\setminus [s]}}} \hat{N}.
\end{align*}
Note that by \Cref{lem:sym_part_tr} we have \begin{align*}
    X \le \brac{1 + O\brac{\frac{t^2}{d}}} \frac{d^{t-s}s!}{t!} \brac{ \id_\reg{B_1} \ot \SymmSub(\reg{B_2^{[s]}}) }, \quad 
    Y \le \brac{1 + O\brac{\frac{t^2}{d}}} \frac{d^{t-s}s!}{t!} \brac{ \id_\reg{C_1} \ot \SymmSub(\reg{C_2^{[s]}}) }.
\end{align*}
Similar to \cref{eq:pt}, we can write
\begin{align*}
  \abs{\Tr\brac{ \tr_{\reg{B_2^{[t]\setminus [s]}}}\hat{M} \cdot \tr_{\reg{C_2^{[t]\setminus [s]}}}(
 \hat{N})^{\top_{\reg{C_1}}}}} &\leq \sum_{i,j=0}^{d^s-1} \Big \| (\sigma^{1/2} \ot \ketbra{i}{i}) X (\id \ot \ketbra{j}{j}) \Big \|_2 \cdot \Big \| (\sigma^{1/2} \ot \ketbra{j}{i}) {Y}^\top  (\id \ot \ketbra{j}{i})  \Big \|_2~, 
\end{align*}
where the term on the right can be bounded as \begin{align*}
    \Big \| (\sigma^{1/2} \ot \ketbra{j}{i}) {Y}^\top  (\id \ot \ketbra{j}{i})  \Big \|_2 &= \sqrt{ \Tr \Big( (\sigma \ot \ketbra{i}{i}) {Y}^\top (\id \otimes \ketbra{j}{j}) \overline{Y}  \Big) } \le \brac{1 + O\brac{\frac{t^2}{d}}} \frac{d^{t-s}s!}{t!},
\end{align*}
since we have $\opnorm{Y^\top} = \opnorm{\overline{Y}} = \opnorm{Y} \le \brac{1 + O\brac{\frac{t^2}{d}}} \frac{d^{t-s}s!}{t!}$ and $\opnorm{\id \ot \ketbraX{j}} = 1$.
Therefore, by Cauchy-Schwarz we have \begin{align*}
    \abs{\Tr\brac{ \tr_{\reg{B_2^{[t]\setminus [s]}}}\hat{M} \cdot \tr_{\reg{C_2^{[t]\setminus [s]}}}(
 \hat{N})^{\top_{\reg{C_1}}}}} &\leq \brac{1 + O\brac{\frac{t^2}{d}}} \frac{d^{t-s}s!}{t!}\sum_{i,j=0}^{d^s-1} \Big \| (\sigma^{1/2} \ot \ketbra{i}{i}) X (\id \ot \ketbra{j}{j}) \Big \|_2 \\
 &\le \brac{1 + O\brac{\frac{t^2}{d}}} \frac{d^{t}s!}{t!}\sqrt{\sum_{i,j=0}^{d^s-1} \Big \| (\sigma^{1/2} \ot \ketbra{i}{i}) X (\id \ot \ketbra{j}{j}) \Big \|_2^2} \\
 &= \brac{1 + O\brac{\frac{t^2}{d}}} \frac{d^{t}s!}{t!}\sqrt{\sum_{i,j=0}^{d^s-1} \tr\brac{ (\sigma \ot \ketbra{i}{i}) X (\id \ot \ketbra{j}{j})X }} \\
 &= \brac{1 + O\brac{\frac{t^2}{d}}} \frac{d^{t}s!}{t!}\sqrt{\sum_{i=0}^{d^s-1} \tr\brac{ X^2(\sigma \ot \ketbra{i}{i}) }} \\
 &\le \brac{1 + O\brac{\frac{t^2}{d}}} \frac{d^{t}s!}{t!} \sqrt{d^s \brac{\frac{d^{t-s}s!}{t!}}^2 } \\
 &= \brac{1 + O\brac{\frac{t^2}{d}}} \frac{d^{2t-s/2}(s!)^2}{(t!)^2},
\end{align*}
where we used the fact that $\opnorm{X} \le \frac{d^{t-s}s!}{t!}$. Plugging everything back into \cref{eq:x}, we get that the distinguishing advantage is bounded by \begin{align*}
    O\brac{\frac{t^2}{d}} + \brac{1+O\brac{\frac{t^2}{d}}}\frac{(t!)^2}{d^{2t}}\sum_{s=1}^t \binom{t}{s}^2 \frac{d^{2t-s/2}(s!)^2}{(t!)^2} &\le O\brac{\frac{t^2}{d}} + \sum_{s=1}^t \brac{\frac{t^2}{\sqrt{d}}}^s \le O\brac{\frac{t^2}{\sqrt{d}}} = O(\eps^2)
\end{align*}
as desired.
\end{proof}

\subsection{Comparison to Non-Haar Distributions}
\noindent We could define other versions of \Cref{conj:shi} tailored to different distributions. However, we prove that the choice of Haar distribution leads to the weakest statement under some natural restrictions. In more detail, we consider the setting where two samples from the distribution are almost orthogonal to each other with high probability. In some sense, this setting is necessary to prove the result below: for instance, suppose the support of $\distr$ is a single state. In that case, simultaneous state indistinguishability for this distribution trivially holds. 

We show the following. 

\begin{claim} \label{clm:haar_reduction}
Suppose there exists a distribution $\distr$ on $\sphere(\C^{2^n})$ and $\eps,\delta,\mu : \N \to \R^+ \cup \{0\}$ such that: \begin{enumerate}
    \item[(1)] \label{item:condition1} $ \pr{ \abs{\braket{\psi}{\psi'}}^2 > \delta \; : \; \ket{\psi}, \ket{\psi'} \from \distr } \le \mu $.
    \item [(2)] \label{item:condition2} $\eps + \mu + \delta \le \negl(n)$.
    \item[(3)] \label{item:condition3} $(\eps,\distr)$-SSI holds.
\end{enumerate}
 
Then, the SHI conjecture (\Cref{conj:shi}) is true.
\end{claim}
\begin{proof}
We prove by contradiction. Suppose $(\eps,\distr)$-simultaneous state indistinguishability holds for some distribution $\distr$ on $\sphere(\C^{2^n})$ and suppose there exists a non-local adversary $\adversary=(\bob,\charlie,\rho)$ that can violate $(\eps,\haar_{2^n})$-simultaneous state indistinguishability.   
\par Consider the following non-local adversary $\adversary' =(\bob',\charlie',\rho')$:
\begin{itemize}
    \item $\rho'=(\rho,(U_{\bfB},U_{\bfC}))$, where $\bob$ and $\charlie$ receive as input $\rho$ and they also additionally receive as input the registers $\reg B$ and $\reg C$, both of which contain the description of the \textbf{same} Haar random unitary $U_{\bfB}=U_{\bfC} \leftarrow \haar(\unitarygp_{2^n})$.
    \item $\bob'$ receives $\ket{\psi_{\bob}}$ from the external challenger and $\charlie'$ receives as input $\ket{\psi_{\charlie}}$ from the external challenger. 
    \item $\bob'$ runs $\bob$ on $\ket{\psi_\bob'} = U_{\bob} \ket{\psi_{\bob}}$. Similarly, $\charlie'$ runs $\charlie$ on $\ket{\psi_\charlie'} = U_{\charlie} \ket{\psi_{\charlie}}$.
    \item $(\bob',\charlie')$ output $(b_{\bob},b_{\charlie})$, where $b_{\bob}$ is the output of $\bob$ and $b_{\charlie}$ is the output of $\charlie$.  
\end{itemize}

Define $\distr^\id, \distr^\ind$ as before. We consider the two cases. If $\ket{\psi_\bob}\ket{\psi_\charlie} \from \distr^\id$, then $\ket{\psi_\bob'}\ket{\psi_\charlie'}$ is distributed according to $\haar_{2^n}^\id$. If $\ket{\psi_\bob}\ket{\psi_\charlie} \from \distr^\id$, then we will show that $\ket{\psi_\bob'}\ket{\psi_\charlie'}$ is close to $\haar_{2^n}^\ind$ in trace distance using condition (1). Let $\sigma$ be the density matrix, and let $\sigma'$ be the reduced density matrix after conditioning on $\abs{\braket{\psi_\bob'}{\psi_\charlie'}}^2 \le \delta$. Let \begin{align*}
    \widetilde{\rho} = \E_{U \from \haar(\unitarygp_{2^n}) } \brac{U \otimes U} \brac{ \ketbraX{0^n} \otimes \ketbraX{1^n} } \brac{U^\dagger \otimes U^\dagger}
\end{align*} and let $\rho_{\haar_{2^n}^\ind}$ be the density matrix representing $\haar_{2^n}^\ind$, then \begin{align*} 
    \tracedist{\sigma}{\rho_{\haar_{2^n}^\ind}} &= \tracedist{\sigma}{\frac{\id}{2^n} \otimes \frac{\id}{2^n}} \le \mu + (1-\mu)\tracedist{\sigma'}{\frac{\id}{2^n} \otimes \frac{\id}{2^n}} \\
    &\le \mu + (1-\mu)\tracedist{\sigma'}{\widetilde{\rho}} + (1 - \mu) \tracedist{\widetilde{\rho}}{\frac{\id}{2^n} \otimes \frac{\id}{2^n}} \\
    &\le \mu + (1 - \mu)\delta + (1 - \mu)\frac{1}{2^n} \\
    &= \mu + \delta + \frac{1}{2^n}.
\end{align*}
\\
The second from last line follows from the monotonicity of trace distance. From the above description, it follows that if $\adversary$ violates $(\eps,\haar_{2^n})$-simultaneous state indistinguishability then $\adversary'$ violates $(\eps + \mu + \delta + 1/2^n,\distr)$-simultaneous state indistinguishability, which is a contradiction. 
\end{proof}

\subsection{Lower Bound} \label{sec:haar_lb}

We demonstrate an unentangled attack which show that the dimension of the Haar random states $(2^n)$ indeed need to be super-polynomial in $n$ for (\Cref{conj:shi}) to be true.

\paragraph{Unentangled Attack.}
In this attack, $\bob$ and $\charlie$ share no state beforehand, and simply each output the result of measuring the first qubit of their register in the computational basis. Defined as such, $(\bob, \charlie, \ketbraX{\bot}$ is a $(1/\exp(n), \haar_{2^n})$-simultaneous state distinguisher. \\

\noindent {\bf Analysis:} If $\bob$ and $\charlie$ receive independent Haar states, then the probability that they output $(0,0)$ equals $1/4$. If they receive identical Haar states, we can calculate the probability by considering the type basis. There are $d(d+1)/2$ types total, and the probability of the first bits being $(0,0)$ is given by \begin{align*}
    \frac{2}{d(d+1)} \brac{ \frac{d}{2} + \binom{d/2}{2} } = \frac{d+2}{4(d+1)}.
\end{align*}
Thus, the distinguishing advantage is \begin{align*}
    \frac{d+2}{4(d+1)} - \frac{1}{4} = \frac{1}{4(d+1)} = O(1/d).
\end{align*}

\noindent We leave it for future work to determine whether this attack can be generalized to the $t$-copy case, achieving an advantage of $O(t^2/d)$.

\newcommand{\qct}{\rho_{\ct}}

\section{Applications}

We present applications of simultaneous Haar indistinguishability (\Cref{sec:haar_indist}) to single-decryptor encryption (\Cref{sec:sde}) and unclonable encryption (\Cref{sec:ue}). Ordinarily, single-decryptor encryption is defined with classical ciphertexts and quantum decryption keys, whereas unclonable encryption is defined using quantum ciphertexts and classical encryption/decryption keys. 
We achieve relaxed notions of both primitives above: namely, we additionally allow quantum ciphertexts in single-decryptor encryption and quantum decryption keys in unclonable encryption. 
\par In~\Cref{sec:secret:sharing}, we show how to construct leakage-resilient quantum secret sharing of classical messages, which additionally guarantees security against an eavesdropper that learns classical leakage of all the shares.

\subsection{Single-Decryptor Encryption with Quantum Cipertexts} \label{sec:sde}
\subsubsection{Definitions}

\noindent We adopt the definition of single-decryptor encryption by~\cite{GZ20} to the setting where the ciphertexts can be quantum. 
 
\begin{definition}[Single-Decryptor Encryption] \label{def:sde}
A \emph{single-decryptor encryption} (SDE) scheme is a tuple of QPT algorithms $(\gen, \enc, \dec)$:
\begin{itemize}
    \item $\gen(1^\secparam)$ takes as input a security paramter. It outputs a classical encryption key $\ekey$ and a one-time quantum decryption key $\qkey$.
    \item $\enc(\ekey, m)$ takes as input an encryption key $\ekey$, and a classical message $m$. It outputs a quantum ciphertext $\qcipher$. We require that $\enc$ is pseudo-deterministic.
    \item $\dec(\qkey, \qcipher)$ takes as input a quantum decryption key $\qkey$, a quantum ciphertext $\qcipher$ and outputs a classical message $m$.
\end{itemize}
\end{definition}

\begin{definition}[Correctness] \label{def:sde_correctness}
    A SDE scheme $(\gen, \enc, \dec)$ with quantum ciphertexts is \emph{correct} if for any security parameter $\secparam$ and any message $m$ we have \begin{align*}
        \pr{ m' = m \; : \; \substack{ (\ekey, \qkey) \from \gen(1^\secparam) \\ \qcipher \from \enc(\ekey, m) \\ m' \from \dec(\qkey, \qcipher)} } \ge 1 - \negl(\secparam).
    \end{align*}
\end{definition}

\noindent Before defining security, we introduce a notation $\enc^T(\ekey, m)$, which means sampling randomness $r$ and running $\enc(\ekey, m; r)$ for $T$ times.

\begin{definition}[Security of SDE] \label{def:sde_security}
    An SDE scheme $(\gen, \enc, \dec)$ is called \emph{(information-theoretically) secure against identical ciphertexts} if for any cloning adversary $\abc$ and any pair of messages $(m_0, m_1)$ we have \begin{align*}
        \pr{ b_\bob = b_\charlie = b \; : \; \substack{(\ekey, \sqkey) \from \gen(1^\secparam) \\ \rho_\reg{B C} \from \alice(\sqkey) \\ b \uniform \bit, \quad  \qcipher^{\ot 2} \from \enc^2(\ekey, m_b)  \\  b_\bob \from \bob(\qcipher, \rho_\reg{B}), \quad b_\charlie \from \charlie(\qcipher, \rho_\reg{C}) } } \le \frac{1}{2} + \negl(\secparam) ,
    \end{align*}
    where $\rho_\reg{E}$ denotes the $\reg{E}$ register of the bipartite state $\rho_\reg{B C}$ for $\reg{E} \in \{\reg{B}, \reg{C}\}$.
\end{definition}

\noindent 
\begin{remark}[Identical Ciphertexts] \label{rem:sde_independent_ch}
    In \Cref{def:sde_security} we consider security against identical ciphertexts. One can similarly define security against ciphertexts that are independently generated. This alternate definition was achieved in the plain model by \cite{AKL23}, and it only requires independent-challenge Goldreich-Levin.
\end{remark}

\begin{remark} \label{rem:identical_ct}
    Note that \Cref{def:sde_security} need not be physical for an arbitrary scheme $(\gen, \enc, \dec)$ without the requirement that $\enc$ is pseudo-deterministic due to the fact that $\ket{\ct}$ may be unclonable even for the encryptor. Nonetheless, our construction satisfies this condition, with the classical randomness of $\enc$ being used for sampling a Haar random state.
\end{remark}

\noindent Below, we consider a stronger security definition where many copies of the quantum ciphertext are given to the adversary. This matches the case of classical ciphertext more closely, since classical ciphertexts can be cloned arbitrarily. Another implication of this stronger definition is that security holds against adversaries who can clone the quantum ciphertexts.

\begin{definition}[$t$-Copy Security of SDE] \label{def:sde_sec_tcopy}
    An SDE scheme $(\gen, \enc, \dec)$ is called \emph{(information-theoretically) $t$-copy secure against identical ciphertexts} if for any cloning adversary $\abc$ and any pair of messages $(m_0, m_1)$ we have \begin{align*}
        \pr{ b_\bob = b_\charlie = b \; : \; \substack{(\ekey, \sqkey) \from \gen(1^\secparam) \\ \rho_\reg{B C} \from \alice(\sqkey) \\ b \uniform \bit, \quad  \qcipher^{\ot 2t} \from \enc^{2t}(\ekey, m_b)  \\  b_\bob \from \bob(\qcipher^{\otimes t}, \rho_\reg{B}), \quad b_\charlie \from \charlie(\qcipher^{\otimes t}, \rho_\reg{C}) } } \le \frac{1}{2} + \negl(\secparam) ,
    \end{align*}
    where $\rho_\reg{E}$ denotes the $\reg{E}$ register of the bipartite state $\rho_\reg{B C}$ for $\reg{E} \in \{\reg{B}, \reg{C}\}$.
\end{definition}

\subsubsection{Construction}    \label{sec:sde_construction}
Let $(\gen_\ue, \enc_\ue, \dec_\ue)$ be a one-time unclonable encryption scheme with $C^n$-weak unclonable security and message space $\cM \subseteq \bit^n$, with $n = \poly(\secparam)$, where $C \in (1/2,1)$ is a constant. Let $\distr=\haar_{2^n}$, so that: \begin{enumerate}
    \item $(\eps, \distr^t)$-SSI holds for some $\eps = \negl(\secparam)$ by \Cref{thm:shi_many_copy} as long as $t^2/2^{n/2}$ is negligible.
    \item $\E_{\ket{\psi} \from \distr} \ketbraX{\psi} = \id_{2^n}$.
\end{enumerate}
We construct SDE for single-bit messages ($m \in \bit$) secure against identical ciphertexts as follows:

\begin{itemize}
    \item $\gen(1^\secparam)$ samples a random message $x \uniform \cM$ and a key $k \from \gen_\ue(1^\secparam)$. It computes $\ket{\psi} \from \enc_\ue(k, x)$. It outputs an encryption key $\ekey = (k,x)$ and a decryption key $\qkey  = \ket{\psi}$.
    \item $\enc(\ekey, m)$ samples $\ket{\varphi} = \sum_{y} \alpha_y \ket{y} \from \distr$. It parses $\ekey = (k,x)$ and computes the state $\ket{\phi} = \sum_y \alpha_y \ket{y}\ket{ \inner{y}{x} \oplus m}$. It outputs a quantum ciphertext $\qcipher = \brac{k, \ket{\phi}}$. Note that $\enc$ is pseudo-deterministic given that it can sample from $\distr$ using classical randomness. 
    \item $\dec(\qkey, \qcipher)$ parses $\qcipher = \brac{k, \ket{\phi}}$. It computes $ x \from \dec_\ue(k, \ket{\phi})$. It computes $U_x \ket{\psi}$ and measures the second register to obtain $m$, where $U_x$ is the unitary defined as $U_x\ket{y}\ket{z} = \ket{y}\ket{z \oplus \inner{y}{x}}$. It outputs $m$. 
\end{itemize}

\begin{remark}  \label{rem:tdesign}
    Since $t$ is bounded (see \Cref{thm:sde_sec_tcopy}), we can use a $2t$-state design to instantiate the Haar random state used in the construction. We can similarly use a $2t$-state design to instantiate our construction of unclonable encryption in \Cref{sec:ue}.
\end{remark}

\noindent From the correctness of $(\gen_\ue, \enc_\ue, \dec_\ue)$, it follows that $\dec$ recovers $m$ (with probability negligibly close to 1) from $\qcipher$, where $\qcipher$ is an encryption of $m$. 

\subsubsection{Security Proof}

We first show a lemma that is needed in our security proof:

\begin{lemma}[Simultaneous Quantum Goldreich-Levin with Correlated Input] \label{lem:gl_cor}
    Let $\abc$ be a cloning adversary that given\footnote{Here $\rho$ is given to $\alice$ and then a bipartite state $\sigma$ is given to $\bob$, $\charlie$ in the challenge phase.} $(\rho, (\sigma, (y_\bob, b_\bob), (y_\charlie, b_\charlie)))$, where $y_\bob, y_\charlie \in \bit^n$ are i.i.d. uniform strings and $b_\bob, b_\charlie \in \bit$ are random bits satisfying $b_\bob \oplus b_\charlie = \inner{y_\bob}{x} \oplus \inner{y_\charlie}{x}$, can output $(\inner{y_\bob}{x}, \inner{y_\charlie}{x})$ with probability at least $1/2 + \eps$. Then, there is an extractor $\abcprime$ that given input $(\rho, \sigma)$ outputs $(x,x)$ (running $\abc$ as a subprotocol) with probability $\poly(\eps)$.
\end{lemma}

\begin{proof}
    Firstly, we will come up with $\abctild$, where given $(\rho, (\sigma, y_\bob, y_\charlie))$ outputs $(b_\bob', b_\charlie')$ satisfying $b_\bob' \oplus b_\charlie' = \inner{y_\bob}{x} \oplus \inner{y_\charlie}{x}$ with probability $1/2 + \eps$. Then, we can apply \Cref{lem:quantumGL} to finish the proof.
    \par Now onto proving that such $\abctild$ exists, indeed, $\alicetild$ is the same as $\alice$. On the other hand, $\bobtild,\charlietild$ will sample $b_\bob, b_\charlie$ independently and run $\bob,\charlie$ as subprotocols, outputting whatever they output. Let us establish some notation before moving forward. Let $E_\bob$ be the event that $\bobtild$ outputs $\inner{y_\bob}{x}$, and similarly $E_\charlie$ be the event that $\charlietild$ outputs $\inner{y_\charlie}{x}$. Let $E$ be the event that $b_\bob \oplus b_\charlie = \inner{y_\bob}{x} \oplus \inner{y_\charlie}{x}$; note that $\pr{E} = 1/2$. We can express the condition on $\abc$ as \begin{align}
        \pr{E_\bob \land E_\charlie \given E} \ge \frac{1}{2} + \eps. \label{eq:goodcase}
    \end{align}
    Additionally, the view of $\bob$ (hence that of $\bobtild$) is independent of $E$, thus due to no-signalling we have \begin{align*}
        \pr{E_\bob \given E} = \pr{E_\bob \given \neg {E}}, \quad \pr{E_\charlie \given E} = \pr{E_\charlie \given \neg {E}};
    \end{align*}
    which imply that \begin{align}
        & \quad \pr{E_\bob \land E_\charlie \given E} + \pr{\neg E_\bob \land \neg E_\charlie \given \neg E} \nonumber \\ &= \brac{\pr{E_\bob \given E} - \pr{E_\bob \land \neg E_\charlie \given E}} + \brac{\pr{\neg E_\bob \given \neg E} - \pr{\neg E_\bob \land E_\charlie \given \neg E}} \nonumber \\
        &= 1 - \pr{E_\bob \land \neg E_\charlie \given E} - \pr{\neg E_\bob \land E_\charlie \given \neg E} \nonumber \\
        &= 1 - \brac{\pr{\neg E_\charlie \given E} - \pr{\neg E_\bob \land \neg E_\charlie \given E}} - \brac{\pr{E_\charlie \given \neg E} - \pr{E_\bob \land E_\charlie \given \neg E}} \nonumber \\
        &= \pr{\neg E_\bob \land \neg E_\charlie \given E} + \pr{ E_\bob \land E_\charlie \given \neg E}. \label{eq:diag_sum}
    \end{align}

Now, using \cref{eq:goodcase,eq:diag_sum}, the probability that the output of $\abctild$ satisfies $b_\bob' \oplus b_\charlie' = \inner{y_\bob}{x} \oplus \inner{y_\charlie}{x}$ can be calculated as \begin{align*}
    &\pr{E_\bob \land E_\charlie} + \pr{\neg E_\bob \land \neg E_\charlie}\\
    &= \frac{1}{2} \brac{ \pr{ E_\bob \land E_\charlie \given E } + \pr{E_\bob \land E_\charlie \given \neg E}} + \frac{1}{2} \brac{ \pr{ \neg E_\bob \land \neg E_\charlie \given E } + \pr{\neg E_\bob \land \neg E_\charlie \given \neg E}} \\
    &= \pr{E_\bob \land E_\charlie \given E} + \pr{\neg E_\bob \land \neg E_\charlie \given \neg E} \ge \pr{E_\bob \land E_\charlie \given E} \ge \frac{1}{2} + \eps.
\end{align*}

\end{proof}

\begin{remark}
\Cref{lem:quantumGL} proves a special case of the simultaneous inner product conjecture postulated in~\cite{AB23}. Roughly speaking, the simultaneous inner product conjecture states that if a set of bipartite states $\{\rho_x\}_{x \in \{0,1\}^n}$ is such that any non-local adversary $(\bob,\charlie)$ given $\rho_x$, where $x \xleftarrow{\$} \{0,1\}^n$, cannot recover $x$ (except with negligible probability) then $\bob$ and $\charlie$ cannot distinguish Goldreich-Levin samples versus uniform samples (except with negligible advantage). The conjecture is parameterized by the distribution of the samples and also by the algebraic field associated with the samples.  \Cref{lem:quantumGL} shows that the conjecture is true for a correlated distribution of samples and when the field in question is $\mathbb{F}_2$. 
\end{remark}

\begin{remark}
    \Cref{lem:quantumGL} resolves an important technical issue we will face in our unclonable security proofs (\Cref{thm:sde_security,thm:ue_qkeys_sec}), similar to the one faced by \cite{KT22}\footnote{See Remark 7 (pp. 53) in \cite{KT22}.}. Namely, $\bob$ and $\charlie$ seem to get additional information about the hidden values $\inner{y_\bob}{x}, \inner{y_\charlie}{x}$ by holding secret shares of their XOR value. Above we show that this is in fact not the case, i.e. the adversary does not get additional power from these shares. In \cite{KT22}, the authors utilized alternative security definitions to overcome this issue.
\end{remark}

\begin{theorem} \label{thm:sde_security}
    $(\gen, \enc, \dec)$ above is secure against identical ciphertexts.
\end{theorem}

\begin{proof}
    We will use the hybrid method. Without loss of generality set $(m_0, m_1) = (0,1)$ and define the folllowing hybrids:
    \begin{enumerate}
        \item {\bf Hybrid 0.} This is the original security experiment described in \Cref{def:sde_security}.
        \item {\bf Hybrid 1.} In this hybrid, instead of giving $\qcipher =\brac{k, \sum_{y} \alpha_y \ket{y}\ket{\inner{y}{x} \oplus m} }$ to both $\bob$ and $\charlie$, we give $\brac{ k, y_\bob, \inner{y_\bob}{x} \oplus m }$ to $\bob$ and $\brac{ k, y_\charlie, \inner{y_\charlie}{x} \oplus m }$ to $\charlie$, where $y_\bob, y_\charlie \in \bit^n$ are uniformly random independent classical strings.
        \item {\bf Hybrid 2.} In this hybrid, we change the victory condition of the adversary $\abc$. Instead of outputting $b=m_b$, we ask that $\bob$ output $\inner{y_\bob}{x}$ and $\charlie$ output $\inner{y_\charlie}{x}$.

        \item {\bf Hybrid 3.} In this hybrid, we only give $k$ to $\bob$ and $\charlie$. We also change the victory condition: in order to succeed, $\bob$ and $\charlie$ both need to output $x$.

    \end{enumerate}

    Let $p_i$ be the optimal probability that a QPT adversary $\abc$ succeeds in \textbf{Hybrid $i$}. We will split the proof into the following claims:

    \begin{claim}   \label{clm:01}
        $\abs{p_0 - p_1} \le \negl(\secparam)$.
    \end{claim}
    \begin{proof}
        Suppose not, we will construct a non-local adversary $(\bob', \charlie', \rho)$ which can simultaneously distinguish $\ket{\varphi} \otimes \ket{\varphi}$ from $\ket{\varphi} \otimes \ket{\varphi'}$, where $\ket{\varphi}, \ket{\varphi'} \from \distr$. Note that by assumption $\ket{\varphi} \otimes \ket{\varphi'}$ is identically distributed to $\ket{y_\bob} \otimes \ket{y_\charlie}$ where $y_\bob, y_\charlie \from \bit^n$. Now onto the construction of $(\bob', \charlie', \rho)$: \begin{itemize}
            \item The mixed state $\rho$ is defined as follows: Sample a random message $x \uniform \cM$ and key $k \from \gen_\ue(1^\secparam)$. Compute $\ket{\psi} \from \enc_\ue(k,x)$ and $\rho' \from \alice(\ket{\psi})$. Finally sample a random bit $m \uniform \bit$ and output $\rho = \rho' \otimes \ket{k,x,m}_\reg{B} \otimes \ket{k,x,m}_\reg{C}$.
            \item $\bob'$ on input $\ket{\phi} = \sum_y \beta_y \ket{y}$ computes the state $\ket{\phi'} = \sum_{y} \beta_y \ket{y} \ket{\inner{y}{x} \oplus m}$. Then, it runs $\bob$ on input $\brac{ k, \ket{\phi'}, \rho_\reg{B} }$ to obtain $b_\bob \in \bit$ and outputs $b_\bob \oplus m$. We similarly define $\charlie'$.
        \end{itemize}
    On input $\ket{\varphi} \otimes \ket{\varphi}$, the probability that $(\bob', \charlie')$ output $(0,0)$ equals $p_0$, whereas on input $\ket{y_\bob} \otimes \ket{y_\charlie}$ the same probability equals $p_1$. Hence, $(\bob', \charlie', \rho)$ is a $\abs{p_0 - p_1}$-simultaneous state distinguisher for $\distr$, which suffices for the proof.
    \end{proof}

    \begin{claim} \label{clm:12}
        $p_1 \le p_2$.
    \end{claim}
    \begin{proof}
        Suppose $\abc$ succeeds in \textbf{Hybrid 1} with probability $p_1$. Construct $\abcprime$ for \textbf{Hybrid 2} as follows: run $\abc$ so that $\bob'$ obtains $b_\bob$ and $\charlie'$ obtains $b_\charlie$. Let $(k,y_\bob,z_\bob)$ and $(k,y_\charlie,z_\charlie)$ be the inputs given to $\bob'$ and $\charlie'$, respectively. $\bob'$ outputs $b_\bob \oplus z_\bob$ and $\charlie'$ outputs $b_\charlie \oplus z_\charlie$. Conditioned on $\abc$ succeeding, $\abcprime$ will also succeed.
    \end{proof}

    \begin{claim}   \label{clm:23}
        If $p_2 > 1/2$, then $p_3 \ge \poly(p_2 - 1/2)$.
    \end{claim}
    \begin{proof}
        Follows from \Cref{lem:gl_cor}, where $\abc$ corresponds to {\bf Hybrid 2} and $\abcprime$ corresponds to {\bf Hybrid 3}. Note that the bits $b_\bob = \inner{y_\bob}{x} \oplus m$ and $b_\charlie = \inner{y_\charlie}{x} \oplus m$ are uniform with correlation $b_\bob \oplus b_\charlie = \inner{y_\bob}{x} \oplus \inner{y_\charlie}{x}$ since $m$ is uniform.
    \end{proof}

    \begin{claim}   \label{clm:3}
        $p_3 \le \negl(\secparam)$.
    \end{claim}

    \begin{proof}
        This follows directly from the weak-UE security of $(\gen_\ue, \enc_\ue, \dec_\ue)$.
    \end{proof}

    \noindent Combining Claims \ref{clm:01},\ref{clm:12},\ref{clm:23},\ref{clm:3}, it follows that $p_0 \le 1/2 + \negl(\secparam)$.
    
\end{proof}

\paragraph{$t$-Copy Security.}
We show that our construction remains secure if up to $O(n)$ copies of the quantum ciphertext is given to the adversary in the unclonable security experiment.

\begin{theorem} \label{thm:sde_sec_tcopy}
    Let $c>0$ be a constant and $t\le (\log_2(1/\sqrt{C}) - c)n)$, then $(\gen, \enc, \dec)$ above is $t$-copy secure against identical ciphertexts.
\end{theorem}

\begin{proof}
We generalize the proof above to the case when the adversary gets $t$ copies (each) of the ciphertext using the $t$-copy version of SSI. We use the hybrid method as before:  \begin{enumerate}
    \item {\bf Hybrid 0.} This is the original security experiment described in \Cref{def:sde_sec_tcopy}.
    \item {\bf Hybrid 1.} In this hybrid, instead of giving $\qcipher^{\otimes t}$, or\footnote{Note that we can compress $t$ copies of $k$ into a single copy since it is classical.} $ \brac{k, \brac{\sum_{y} \alpha_y \ket{y}\ket{\inner{y}{x} \oplus m} }^{\otimes t}}$, to both $\bob$ and $\charlie$, we give $ \brac{k, \brac{\sum_{y} \alpha_y \ket{y}\ket{\inner{y}{x} \oplus m} }^{\otimes t}}$ to $\bob$ and $ \brac{k, \brac{\sum_{y} \alpha'_y \ket{y}\ket{\inner{y}{x} \oplus m} }^{\otimes t}}$ to $\charlie$, where $\sum_y \alpha_y \ket{y} \from \distr$ and $\sum_y \alpha_y' \ket{y} \from \distr$ are independently generated.
    \item {\bf Hybrid 2.} In this hybrid, we instead give $\brac{ k,\brac{ y_\bob^i, \inner{y_\bob^i}{x} \oplus m}_{i\in [t]} }$ to $\bob$ and $\brac{ k,\brac{ y_\charlie^j, \inner{y_\charlie^j}{x} \oplus m}_{j\in [t]} }$ to $\charlie$, where $y_\bob^i, y_\charlie^j \in \bit^n$ are uniformly random independent classical strings conditioned on $\brac{y_\bob^i}_{i \in [t]}$ being $t$ distinct strings and $(y_\charlie^j)_{j \in [t]}$ being $t$ distinct strings.
    
    \item {\bf Hybrid 3.} In this hybrid, we remove the restriction that $\brac{y_\bob^i}_{i \in [t]}$ and $(y_\charlie^j)_{j \in [t]}$ are distinct.
    \item {\bf Hybrid 4.} In this hybrid, we instead give $\brac{k, y_\bob^1, \inner{y_\bob^1}{x} \oplus m, \brac{ y_\bob^i, \inner{y_\bob^i}{x}}_{i\in [t]\setminus[1]} }$ to $\bob$ and \\ $\brac{k, y_\charlie^1, \inner{y_\charlie^1}{x} \oplus m, \brac{ y_\charlie^i,  \inner{y_\charlie^i}{x}}_{i\in [t]\setminus[1]} }$ to $\charlie$, where $\brac{y_\bob^i}_{i \in [t]}$ and $(y_\charlie^j)_{j \in [t]}$ are uniformly random strings.
    \item {\bf Hybrid 5.} In this hybrid, we give $\brac{k, \brac{ y_\bob^i, \inner{y_\bob^i}{x}}_{i\in [t]\setminus[1]} }$ to $\bob$ and $\brac{k, \brac{ y_\charlie^i, \inner{y_\charlie^i}{x}}_{i\in [t]\setminus[1]} }$ to $\charlie$. We also change the victory condition by requiring that $\bob$ and $\charlie$ output $(x,x)$.
\end{enumerate}
Let $p_i$ be the optimal probability that a cloning adversary $\abc$ succeeds in \textbf{Hybrid $i$}. We will split the proof into the following claims:

\begin{claim}   \label{clm:01_t}
        $\abs{p_0 - p_1} \le \negl(\secparam)$.
\end{claim}
\begin{proof}
    Follows from $(\eps, \distr^t)$-SSI, and the proof is nearly identical to that of \Cref{clm:01}.
\end{proof}
\begin{claim}   \label{clm:12_t}
        $p_2 \ge p_1 - \negl(\secparam) $.
\end{claim}
\begin{proof}
    Suppose $\abc$ succeeds in {\bf Hybrid 1} with probability $p_1$. We will describe $\abcprime$ which succeeds with probability $p_2 - \negl(\secparam)$ in {\bf Hybrid 2}. Note that the only difference between the hybrids is the input given to $(\bob, \charlie)$ or $(\bobprime,\charlieprime)$. Therefore, it suffices to set $\aliceprime = \alice$ and show that $\bobprime$ (resp., $\charlieprime$) can transform his input to what $\bob$ (resp., $\charlie$) receives in {\bf Hybrid 1} up to negligible trace distance. Indeed, given $\brac{ y_\bob^i, \inner{y_\bob^i}{x} \oplus m}_{i\in [t]}$ with $y_\bob^1 \ne y_\bob^2 \ne \dots \ne y_\bob^t$, $\bobprime$ records $\brac{i, y_\bob^i}_{i \in [t]}$ and coherently applies a random index-wise permutation $\sigma \in S_t$, resulting in the state \begin{align*}
        \frac{1}{t!}\sum_{\sigma \in S_t} \ket{\sigma} \bigotimes_{i=1}^t \ket{y_\bob^{\sigma(i)}}\ket{\inner{y_\bob^{\sigma(i)}}{x} \oplus m }.
    \end{align*}
    Then using the table he recorded $\bobprime$ uncomputes and discards the first register, which results in the state \begin{align*}
        \frac{1}{t!}\sum_{\sigma \in S_t} \bigotimes_{i=1}^t \ket{y_\bob^{\sigma(i)}}\ket{\inner{y_\bob^{\sigma(i)}}{x} } = U_x^{\ot t} \ket{\type(y_\bob^1,\dots,y_\bob^t)} \ket{m}.
    \end{align*}
    Over the expectation ${(y_\bob^1,\dots,y_\bob^t) \from \distinct_{t,2^s}} $, the mixed state (by \Cref{lem:distinct_type}) is negligibly close to \begin{align*}
        U_x^{\ot t} \E_{\ket{\varphi} \from \distr} \brac{\varphi \ot \ketbraX{m}}^{\ot t} {U_x^\dagger}^{\ot t},
    \end{align*}
    which is the state $\bob$ receives in {\bf Hybrid 1}. Arguing similarly for $\charlieprime$, the claim follows.
\end{proof}

\begin{claim}   \label{clm:23_t}
        $\abs{p_3 - p_2} \le \negl(\secparam) $.
\end{claim}
\begin{proof}
    By collision bound, the difference is bounded by $O(t^2/2^n)$ which is negligible.
\end{proof}

\begin{claim}   \label{clm:34_t}
        $p_4 \ge p_3$.
\end{claim}
\begin{proof}
    Let $\abc$ succeed in {\bf Hybrid 3} with probability $p_3$. We construct $\abcprime$ that succeeds with the same probability $p_3$ in {\bf Hybrid 4}: \begin{itemize}
        \item $\aliceprime$ is the same as $\alice$.
        \item $\bobprime$ is given $\brac{k, y_\bob^1, \inner{y_\bob^1}{x} \oplus m, \brac{ y_\bob^i, \inner{y_\bob^i}{x}}_{i\in [t]\setminus[1]} }$. For $i \in [t] \setminus [1]$ he computes the XOR of $(y_\bob^1, \inner{y_\bob^1}{x} \oplus m)$ and $(y_\bob^i, \inner{y_\bob^i}{x})$ to get $(y_\bob^1 \oplus y_\bob^i, \inner{y_\bob^1 \oplus y_\bob^i}{x} \oplus m) = (z_\bob^i, \inner{z_\bob^i}{x} \oplus m)$, where $z_\bob^i := y_\bob^1 \oplus y_\bob^i$. Also set $z_\bob^1:= y_\bob^1$. Then $\bobprime$ runs $\bob$ on input $\brac{ k,\brac{ z_\bob^i, \inner{z_\bob^i}{x} \oplus m}_{i\in [t]} }$ and outputs the answer.
        \item $\charlieprime$ is defined similarly.
    \end{itemize}
    Since $\brac{z_\bob^i}_{i \in [t]}$ are uniformly random and independent, the view of $\bob$ is the same as that in {\bf Hybrid 3}. Arguing similarly for $\charlie$, it follows that $\abcprime$ succeeds with probability $p_3$.
\end{proof}

\begin{claim}   \label{clm:45_t}
    If $p_4 > 1/2$, then $p_5 \ge \poly(p_4 - 1/2)$.
\end{claim}

\begin{proof}
    Follows from \Cref{lem:gl_cor}, where $\abc$ corresponds to {\bf Hybrid 4} and $\abcprime$ corresponds to {\bf Hybrid 5}. Note that the input $\sigma$ received by $(\bob, \charlie)$ (and likewise by $(\bobprime, \charlieprime)$) equals the mixed state defined by $\brac{k, \brac{ y_\bob^i, \inner{y_\bob^i}{x}}_{i\in [t]\setminus[1]} }$ and $\brac{k, \brac{ y_\charlie^i, \inner{y_\charlie^i}{x}}_{i\in [t]\setminus[1]} }$.
\end{proof}

\begin{claim}   \label{clm:5_t}
    $p_5 \le \negl(\secparam)$.    
\end{claim}

\begin{proof}
    Suppose $\abc$ succeeds in {\bf Hybrid 5} with probability $p_5$. We will construct $\abcprime$ against the weak unclonable security of $(\gen_\ue, \enc_\ue, \dec_\ue)$ as follows: \begin{itemize}
        \item $\aliceprime$ is the same as $\alice$.
        \item $\bobprime$ receives $k$ as input in the challenge phase. He samples\footnote{We note that this idea was also used in \cite{CG23}.} $y_\bob^i \uniform \bit^n$ and $b_\bob^i$ for $i = 2,\dots,t$. Then it runs $\bob$ on input $\brac{k, \brac{y_\bob^i, b_\bob^i}_{i \in [t]\setminus [1]}}$ and outputs the answer. 
        \item $\charlieprime$ is defined similarly.
    \end{itemize}
    Note that conditioned on the event that $b_\bob^i = \inner{y_\bob^i}{x}$ and $b_\charlie^i = \inner{y_\charlie^i}{x}$ for all $i \in [t]\setminus [1]$, $\abcprime$ perfectly imitates the view of $\abc$ and hence succeeds with probability $p_5$. Thus, by the security of $(\gen_\ue, \enc_\ue, \dec_\ue)$ we have $2^{-2(t-1)}p_5 \le C^n \implies p_5 < 2^{2t}C^n =4^{t-\log(1/\sqrt{C})} \le 4^{-cn} \le \negl(\secparam). $
\end{proof}
\noindent Combining Claims \ref{clm:01_t}, \ref{clm:12_t}, \ref{clm:23_t}, \ref{clm:34_t}, \ref{clm:45_t}, and \ref{clm:5_t}, we have showed that $p_0 \le 1/2 + \negl(\secparam)$ as desired.
\end{proof}

By instantiating $(\gen_\ue, \enc_\ue, \dec_\ue)$ with the construction of \cite{BL20}, we can set $C = 0.86$ and thus $t = n/10$. Therefore, by setting $n=10t$ we get the following corollary:

\begin{corollary}    \label{cor:sde}
    For any $t = \poly(\secparam)$, there exists a single-decryptor encryption scheme with quantum ciphertexts $t$-copy secure against identical ciphertexts (\Cref{def:sde_sec_tcopy}) in the plain model.
\end{corollary}

\begin{remark}[Bounded Number of Copies]   \label{rem:sde_bounded}
    Our construction is not $t$-copy secure against identical ciphertexts if $t$ is an unbounded polynomial due to a simple attack that measures each copy of the ciphertext in the computational basis. Every measurement (except the first) will give a random linear constraint on the bits of $x$, hence a linear number of copies on average suffice to solve for $x$ entirely. Nonetheless, we can set $n$ accordingly for any fixed polynomial $t$.
\end{remark}
\newcommand{\pirexp}[2]{\mathsf{PirExp}_{#1}\brac{ 1^\secparam, #2 }}
\newcommand{\pirexpT}[2]{\mathsf{PirExp}_{#1}^t\brac{ 1^\secparam, #2 }}
\newcommand{\ueq}{\mathsf{UE}^\mathsf{Q}}
\newcommand{\referee}{\mathcal{R}}

\subsection{Unclonable Encryption with Quantum Keys} \label{sec:ue}
Using our ideas in \Cref{sec:sde}, we construct unclonable encryption with quantum keys in the plain model. Besides allowing the decryption key to be quantum, we do not relax the syntax of unclonable encryption. In particular, we achieve identical-challenge security with quantum challenges.
\subsubsection{Definitions}

\noindent We consider a relaxation of the definitions of unclonable encryption considered by~\cite{BL20,AK21}. 

\begin{definition}[UE with Quantum Decryption Keys]  \label{def:ue_qkeys}
    An \emph{unclonable encryption scheme with quantum decryption keys} is a triplet $(\gen, \enc, \dec)$ of QPT algorithms with the following syntax: \begin{itemize}
        \item $\gen(1^\secparam)$ on input a security parameter in unary outputs a classical encryption key $\ekey$ and a quantum decryption key $\qkey$. We require that $\gen$ is pseudo-deterministic.
        \item $\enc(\ekey, m)$ on input a classical encryption key $\ekey$ and a classical message $m$ outputs a quantum ciphertext $\qcipher$.
        \item $\dec(\qkey, \qcipher)$ on input a quantum decryption key $\qkey$ and a quantum ciphertext outputs a classical message $m$.
    \end{itemize}
\end{definition}

\noindent See \Cref{rem:identical_ct} regarding the pseudo-determinism requirement for $\gen$.

\begin{remark}[Asymmetric secret keys.]
    Note that the syntax above is atypical in the sense that it is a secret-key encryption scheme with asymmetric encryption/decryption keys. The advantage of our syntax is that only the decryption key has to be quantum, while the encryption key can remain classical. We leave it for future work to consider the symmetric-key variant of this primitive with quantum encryption keys. 
\end{remark}

\begin{definition}[Correctness] \label{def:ue_qkeys_correctness}
    An unclonable encryption scheme $(\gen, \enc, \dec)$ with quantum decryption keys is \emph{correct} if for any security parameter $\secparam$ and any message $m$ we have \begin{align*}
        \pr{ m' = m \; : \; \substack{ (\ekey, \qkey) \from \gen(1^\secparam) \\ \qcipher \from \enc(\ekey, m) \\ m' \from \dec(\qkey, \qcipher)} } \ge 1 - \negl(\secparam).
    \end{align*}
\end{definition}

\noindent Below we use the notation $\gen^T(1^\secparam)$ defined similarly as in \Cref{sec:sde}.

\begin{definition}[Unclonable Security] \label{def:ue_qkeys_security}
    An unclonable encryption scheme $\ueq = (\gen, \enc, \dec)$ with quantum decryption keys satisfies \emph{unclonable security} if for any cloning adversary $\abc$ we have 
    
    \begin{align*}
        \pr{b_\bob = b_\charlie = b \; : \; \substack{ (m_0, m_1) \from \alice(1^\secparam) \\ (\ekey, \qkey^{\ot 2}) \from \gen^2(1^\secparam), \quad b \uniform \bit, \quad \qcipher \from \enc(\ekey, m_b) \\ \rho_\reg{BC} \from \alice(\qcipher) \\ b_\bob \from \bob(\qkey, \rho_\reg{B}), \quad b_\charlie \from \charlie(\qkey, \rho_\reg{C}) }},
    \end{align*}
    where $\rho_\reg{E}$ denotes the $\reg{E}$ register of the bipartite state $\rho_\reg{B C}$ for $\reg{E} \in \{\reg{B}, \reg{C}\}$.
\end{definition}

\noindent Similar to \Cref{def:sde_sec_tcopy}, we can define the $t$-copy security of UE with quantum keys, which we formally state below. The only difference is that the adversary will get $t$ copies of the decryption key in the verification phase rather than a single copy.

\begin{definition}[$t$-Copy Unclonable Security] \label{ue_qkeys_sec_tcopy}
    An unclonable encryption scheme $\ueq = (\gen, \enc, \dec)$ with quantum decryption keys satisfies \emph{$t$-copy unclonable security} if for any cloning adversary $\abc$ we have 
   
    \begin{align*}
        \pr{b_\bob = b_\charlie = b \; : \; \substack{ (m_0, m_1) \from \alice(1^\secparam) \\ (\ekey, \qkey^{\ot 2t}) \from \gen^{2t}(1^\secparam), \quad b \uniform \bit, \quad \qcipher \from \enc(\ekey, m_b) \\ \rho_\reg{BC} \from \alice(\qcipher) \\ b_\bob \from \bob(\qkey^{\ot t}, \rho_\reg{B}), \quad b_\charlie \from \charlie(\qkey^{\ot t}, \rho_\reg{C}) }},
    \end{align*}
    where $\rho_\reg{E}$ denotes the $\reg{E}$ register of the bipartite state $\rho_\reg{B C}$ for $\reg{E} \in \{\reg{B}, \reg{C}\}$.
\end{definition}

\subsubsection{Construction} Let $(\gen_\ue, \enc_\ue, \dec_\ue)$ with message space $\cM \subseteq \bit^n$ and $\distr = \haar_{2^n}$ be defined as in \Cref{sec:sde_construction}. We construct unclonable encryption with quantum decryption keys $\ueq = (\gen, \enc, \dec)$ for single-bit messages as follows:

\begin{itemize}
    \item $\gen(1^\secparam):$ on input security parameter $\secparam$, do the following:
    \begin{itemize}
        \item Compute $k \from \gen_\ue(1^\secparam)$,
        \item Sample $x \xleftarrow{\$} \cM$,
        \item Compute $\ket{\psi} \leftarrow \distr$, 
        \item Sample $\widetilde{b} \xleftarrow{\$} \{0,1\}$,
        \item Compute $\ket{\phi} = U_x \ket{\psi}\ket{\widetilde{b}}$, where $U_x \ket{y}\ket{z} = \ket{y}\ket{\inner{y}{x} \oplus z}$. 
        
    \end{itemize}
    \noindent Output the decryption key $\qkey = (k,\ket{\phi})$ and the encryption key $\ekey=(k, x,\widetilde{b})$.

    \item $\enc(\ekey, m):$ on input $\ekey=(k,x, \widetilde{b})$ and a message $m \in \{0,1\}$, do the following: 
    \begin{itemize}
        \item Compute $\qct \leftarrow \enc_{\ue}(k,x)$.
    \end{itemize}
    Output $\ct=(\qct, \widetilde{b} \oplus m)$
    \item $\dec(\qkey, {\ct}):$ on input $\qkey = (k,\ket{\phi})$ and $\ct=(\qct, b')$, compute $x \from \dec_\ue(k, \qct)$. Then compute $U_x\ket{\phi}$ and measure the second register to obtain $\widetilde{b}$. Output $\widetilde{b} \oplus b'$.
\end{itemize}

\paragraph{Security.}
Security of this scheme can be seen to be equivalent to the security of the SDE construction in the previous section. We give the formal details below.

\begin{theorem} \label{thm:ue_qkeys_sec}
    The construction $\ueq = (\gen, \enc, \dec)$ above satisfies correctness and unclonable security.
\end{theorem}

\begin{proof}
    Correctness is trivial to check. We show security by defining a hybrid experiment: \begin{enumerate}
        \item {\bf Hybrid 0.} 
        
        This is the original piracy experiment defined in \Cref{def:ue_qkeys_security}.
        We can set $m_0 = 0, m_1 = 1$ since the message-length is 1.
        \item {\bf Hybrid 1.}
        In this hybrid, we only give $\qcipher$ to $\alice$ instead of $(\qcipher, \widetilde{b} \oplus m)$. In addition, we change the victory condition for $\bob$ and $\charlie$: they now each need to output $\widetilde{b}$ instead of $m$.
        
    \end{enumerate}
    Let $p_i$ be the optimal winning probability of a cloning adversary in {\bf Hybrid $\mathbf{i}$}. Observe that {\bf Hybrid 1} is a relabeling of {\bf Hybrid 0} in the proof of \Cref{thm:sde_security} (i.e. the SDE security experiment), so that $p_1 \le 1/2 + \negl(\secparam)$. Thus, it suffices to show $p_0 \le p_1$, which we do below:
    \par Let $\abc$ be an adversary that succeeds in {\bf Hybrid 0} with probability $p_0$. We will describe $\abcprime$ which succeeds in {\bf Hybrid 1} with probability $p_0$: \begin{itemize}
        \item $\aliceprime$ receives $\qcipher$ from the referee. She samples a random bit $b_1 \in \bit$ and runs $\alice$ on input $(\qcipher, b_1)$.
        \item $\bobprime$ runs $\bob$ on his input to obtain $m_B \in \bit$. He outputs $m_B \oplus b_1$.
        \item Similarly $\charlieprime$ runs $\charlie$ to obtain $m_C \in \bit$ and outputs $m_C \oplus b_1$.
    \end{itemize} 
    By looking at the view of $\abc$ we see that $\abcprime$ succeeds if and only if $\abc$ succeeds.
    
\end{proof}

\begin{remark}
    Our reduction above closely follows \cite{GZ20} who show an equivalence between single-decryptor encryption and unclonable encryption.
\end{remark}

\noindent Combining the proofs of \Cref{thm:sde_sec_tcopy,thm:ue_qkeys_sec}, we get many copy unclonable security for our construction, which we formally state below:

\begin{theorem} \label{thm:ue_qkeys_sec_tcopy}
    Let $c>0$ be a constant and $t \le (\log_2(1/\sqrt{C})-c)n$, then $\ueq = (\gen, \enc, \dec)$ above satisfies $t$-copy unclonable security.
\end{theorem}
\begin{corollary}    \label{cor:ue}
    For any $t = \poly(\secparam)$, there exists an unclonable encryption scheme with quantum decryption keys in the plain model which is $t$-copy unclonable secure.(\Cref{def:sde_sec_tcopy}).
\end{corollary}

\newcommand{\share}{\mathsf{Share}}
\newcommand{\recons}{\mathsf{Rec}}
\newcommand{\leak}{\mathsf{Leak}}

\subsection{Secret Sharing}
\label{sec:secret:sharing}

\paragraph{Definition.} Below we formally define a quantum secret sharing scheme for a single-bit message $m \in \bit$, which is secure against an eavesdropper if only classical strings are leaked from the quantum shares.
\begin{definition}[2-out-of-$n$ Classical-Leakage-Resilient Quantum Secret Sharing] \label{def:ss}
    Let $n =\poly(\secparam)$. A \emph{2-out-of-$n$ classical-leakage-resilient quantum secret sharing scheme} is a tuple of algorithms $(\share, \recons)$ with the following syntax: \begin{itemize}
        \item $\share(1^\secparam,m)$ takes as input a security parameter and a message $m \in \bit$; it outputs a product state $\ket{\psi_m} = \bigotimes_{i=1}^n \ket{\psi_m^i}$ over registers $\reg{S_1,S_2,\dots,S_n}$ with $\poly(\secparam)$ qubits each. 
        \item $\recons(i,j,{\psi^i},{\psi^j})$ takes as input two indices $i,j \in [n]$ and quantum shares $\ket{\psi^i}, \ket{\psi^j}$ over registers $\reg{S_i,S_j}$. It outputs a message $m \in \bit$.
    \end{itemize} 
    It satisfies the following properties:
    \begin{enumerate}
        \item {\bf $\delta$-Correctness:} For all $m \in \bit, \secparam \in \N, (i,j) \in [n] \times [n]$ we have \begin{align*}
            \pr{ m \from \recons(i,j,{\psi_m^i},{\psi_m^j}) \; : \; \ket{\psi} = \bigotimes_{i=1}^n \ket{\psi_m^i} \from \share(1^\secparam, m) } \ge \delta.
        \end{align*}
        \item {\bf Perfect Secrecy:} For all $i \in [n]$, we have \begin{align*}
            \E_{\ket{\psi_0} \from \share(1^\secparam,0)} \ketbraX{\psi_0^i} = \E_{\ket{\psi_1} \from \share(1^\secparam,1)} \ketbraX{\psi_1^i}.
        \end{align*}
        \item {\bf $\ell$-bit $\eps$ Classical-Leakage-Resilience:} For any quantum algorithms $\leak_1,\dots,\leak_n$ that output $\ell$ classical bits, any $n$-partite state $\rho$, and any distinguisher $\alice$, we have \begin{align*}
            \abs{\pr{ 1 \from \alice(y_1,\dots,y_n) \; : \; \substack{{\psi_0} \from \share(1^\secparam, 0) \\ y_i\from \leak_i({\psi_0^i}, \rho_{i}), \; i \in [n]} } - \pr{ 1 \from \alice(y_1,\dots,y_n) \; : \; \substack{{\psi_1} \from \share(1^\secparam, 1) \\ y_i\from \leak_i({\psi_1^i}, \rho_{i}), \; i \in [n]} }} \le \eps,
        \end{align*}
        where $\rho_{i}$ is the $i$-th register of $\rho$.
    \end{enumerate}
\end{definition}

\begin{remark}
Note that we consider the perfect secrecy and the classical leakage resilience properties separately. It is natural to ask if it is possible to combine the two properties into a stronger property that guarantees security even against adversaries who in addition to receiving one of the shares, also receives as input classical leakage on the rest of the shares. Unfortunately, this stronger property cannot be satisfied due to a simple attack via quantum teleportation. Thus, we need to consider these two properties separately.   
\end{remark}

\begin{remark}
The above notion can be generalized in many ways. Firstly, for simplicity, we define the shares to be pure states and we could consider a general notion where the shares could be entangled with each other. Secondly, we can consider sharing schemes guaranteeing security against different adversarial access structures. 
\end{remark}

\subsubsection{Construction}
\noindent We will construct the primitive above using any distribution $\distr$ that satisfies $(\eps, \distr^t, \ell)$-SSI for $\eps = \negl(\secparam)$ as well as the first bullet of \Cref{clm:haar_reduction}. $\distr$ can be instantiated with the Haar distribution $\haar_d$ for $d=2^{4\ell + \secparam}$ by \Cref{cor:ssi_many_party} and \Cref{thm:shi_many_copy}, and $\haar_d$ can in turn be instantiated using a $nt$-design. We will pick $\omega(\log \secparam) \le t \le \poly(\secparam)$, so that the construction is efficient.

\begin{itemize}
    \item $\share(1^\secparam, m)$: Set $d = 2^{4n\ell + \secparam}$ and $\distr = \haar_d$. If $m=0$, it samples $nt$ copies of $\ket{\psi} \from \distr$ and outputs $\ket{\psi}^{\ot t}_{S_1} \ot \dots \ot \ket{\psi}^{\ot t}_{S_n}$. If $m=1$, it independently samples $t$ copies of $\ket{\psi_i} \from \distr$ for $i \in [n]$ and outputs $\ket{\psi_1}^{\ot t}_{S_1} \ot \dots \ot \ket{\psi_n}^{\ot t}_{S_n}$.
    \item $\recons(i,j,\ket{\varphi}_{S_i},\ket{\varphi'}_{S_j})$: It parses $\reg{S_i} = \reg{S_i^{(1)}\dots S_i^{(t)}}$ and $\reg{S_j} = \reg{S_j^{(1)}\dots S_i^{(t)}}$ as $t$ registers. It applies a SWAP test to the registers $\reg{S_i^{(k)}}$ and $\reg{S_j^{(k)}}$ for $k\in [t]$. It outputs $0$ if at least $\lfloor 3t/4 \rfloor$ of the SWAP tests succeed and outputs $1$ otherwise.
\end{itemize}

\subsubsection{Correctness and Secrecy} The construction above has $\delta$-correctness for $\delta \ge 1 - \negl(\secparam)$. Note that for $m=0$ each SWAP test will succeed with probability $1$, so $\recons$ will always output the correct message. If $m=1$ on the other hand, each SWAP test will succeed with probability negligibly close to $1/2$, so by a Chernoff bound $\recons$ will output $0$ with only negligible probability since $t = \omega(\log \secparam)$.
\par Perfect secrecy is trivial given that each share is distributed according to $\distr$.
\subsubsection{Classical Leakage Resilience.}
By \Cref{cor:ssi_many_party} and \Cref{thm:shi_many_copy}, the total variation distance between the output distributions of any $\ell$-bit leakage functions $\leak_1,\dots,\leak_n$ with respect to shares of $m=0$ and $m=1$ is at most \begin{align*}
    O\brac{ \frac{2^{2n\ell}n^3t^2}{\sqrt{d}} } = O\brac{ \frac{n^3t^2}{2^{\secparam/2}} } \le \negl(\secparam).
\end{align*}

\subsection*{Acknowledgements}
The authors would like to thank Ludovico Lami for answering many questions related to non-local state discrimination. PA and FK are supported by the National Science Foundation under Grant No. 2329938. HY is supported by AFOSR awards
FA9550-21-1-0040 and FA9550-23-1-0363, NSF CAREER award CCF-2144219, NSF award CCF-2329939, and the Sloan Foundation.

\bibliographystyle{alpha}
\bibliography{refs}

\newpage 

\appendix
\section{Additional Proofs} \label{sec:additional_proofs}
\subsection{Proof of Simultaneous Quantum Goldreich-Levin}
Below we prove \Cref{lem:quantumGL} by repeating the proof of \cite{AKL23} and observing that the weaker assumption is sufficient. Quantum Goldreich-Levin was originally shown by \cite{AC02} and later by \cite{CLLZ21} for auxiliary input.
\begin{proof}[Proof of \Cref{lem:quantumGL}] 
We can assume that the state $\rho \otimes \sigma = \ketbraX{\varphi_{x}}$ (quantum part of the input of $\alice$) is a pure state, for the mixed state case follows by convexity. By the deferred measurement principle, we can model $\bob, \charlie$ as unitary maps $(U^{r}_B, U^{r'}_C)$ that depend on the classical input $(r,r')$, each followed by a measurement. Below, we will track the registers belonging to $\bob (\bobprime)$ and/or $\charlie (\charlieprime)$ with subscripts $\reg{B},\reg{C}$. The resulting state then is given by \begin{align*} \left( U^{r}_B \otimes U^{r'}_C \right) \ket{\varphi_{x}}_\reg{BC}
&= \bigg( \alpha_{x,r,r'} \ket{\inner{r}{x}}_\reg{B} \ket{\inner{r'}{x}}_\reg{C} \ket{\phi_{x,r,r'}^0}_\reg{BC}
+ \beta_{x,r,r'} \ket{\inner{r}{x}}_\reg{B} \ket{\overline{\inner{r'}{x}}}_\reg{C} \ket{\phi_{x,r,r'}^1}_\reg{BC} \\
&+ \theta_{x,r,r'} \ket{\overline{\inner{r}{x}}}_\reg{B} \ket{\inner{r'}{x}}_\reg{C} \ket{\phi_{x,r,r'}^2}_\reg{BC}
+ \gamma_{x,r,r'} \ket{\overline{\inner{r}{x}}}_\reg{B} \ket{\overline{\inner{r'}{x}}}_\reg{C} \ket{\phi_{x,r,r'}^3}_\reg{BC} \bigg) \\
&=: \ket{\Gamma_{x,r,r'}}, \end{align*}
where $\ket{\phi^j_{x,r,r'}}$ is a normalized state for $j \in \bracC{0,1,2,3}$ and $\alpha_{x,r,r'}, \gamma_{x,r,r'}$ are the coefficients corresponding to the case of $\alice$ succeeding, so that we can express the assumption as \begin{align*} \E_{x,r,r'} \abs{\alpha_{x,r,r'}}^2 + \abs{\gamma_{x,r,r'}}^2 \ge \frac{1}{2} + \varepsilon
\end{align*}
and hence \begin{align}
    &\E_{x,r,r'} \abs{\alpha_{x,r,r'}}^2 - \abs{\beta_{x,r,r'}}^2 - \abs{\theta_{x,r,r'}}^2 + \abs{\gamma_{x,r,r'}}^2 = \E_{x,r,r'} \brac{2\abs{\alpha_{x,r,r'}}^2 + 2\abs{\gamma_{x,r,r'}}^2 - 1} \ge 2\varepsilon \label{eq:coeff_bound}
\end{align}

We now describe the extractor $\aliceprime = (\bobprime, \charlieprime, \rho)$: \begin{itemize}
    \item After $(\bobprime, \charlieprime)$ receives $\ket{\varphi_x}$ as overall input, $\bobprime$ prepares a uniform superposition over $r \in \cR$ and applies the unitary $U_B$, where we define $U_E$ as $U_E \ket{r}\ket{\varphi} = \ket{r}U^{r}_E \ket{\varphi}$ for $E \in \bracC{B,C}$. Then, $\bob'$ applies a $Z$ gate to the register storing the inner product $\inner{r}{x}$, and applies $U_B^\dagger$ to its state. Finally, $\bob'$ measures the register storing the random coins $r$ in the Fourier basis and outputs the result.
    \item $\charlie'$ is defined in a similar fashion.
\end{itemize}
Next, we will analyze the evolution of the state shared by $\bob'$ and $\charlie'$ step by step. Since the actions of $\bob'$ and $\charlie'$ commute, we can synchronously track their operations. After the first step, the state is given by \begin{align*} \left(U_B \otimes U_C \right) \frac{1}{{|\cR|}} \sum_{r,r' \in \cR} \ket{r}_\reg{B}\ket{r'}_\reg{C} \ket{\varphi_{x}}_\reg{BC} =  \frac{1}{{|\cR|}} \sum_{r,r' \in \cR} \ket{r}_\reg{B}\ket{r'}_\reg{C} \ket{\Gamma_{x,r,r'}}. 
\end{align*}

Next, $\bob'$ and $\charlie'$ each apply a $Z$ gate to their register storing the inner product, which results in the state \begin{align*} &\frac{1}{{|\cR|}} \sum_{r,r' \in \cR}  \ket{r}_\reg{B}\ket{r'}_\reg{C} (-1)^{\inner{r}{x} \oplus \inner{r'}{x}} \bigg( \alpha_{x,r,r'} \ket{\inner{r}{x}}_\reg{B} \ket{\inner{r'}{x}}_\reg{C} \ket{\phi_{x,r,r'}^0}_\reg{BC} \\
&- \beta_{x,r,r'} \ket{\inner{r}{x}}_\reg{B} \ket{\overline{\inner{r'}{x}}}_\reg{C} \ket{\phi_{x,r,r'}^1}_\reg{BC}
- \theta_{x,r,r'} \ket{\overline{\inner{r}{x}}}_\reg{B} \ket{\inner{r'}{x}}_\reg{C} \ket{\phi_{x,r,r'}^2}_\reg{BC} \\
&+ \gamma_{x,r,r'} \ket{\overline{\inner{r}{x}}}_\reg{B} \ket{\overline{\inner{r'}{x}}}_\reg{C} \ket{\phi_{x,r,r'}^3}_\reg{BC} \bigg) \\
&=: \frac{1}{{|\cR|}} \sum_{r,r' \in \cR} \ket{r}_\reg{B}\ket{r'}_\reg{C} \ket{\Gamma'_{x,r,r'}},
\end{align*}
with \begin{align*} \inner{\Gamma_{x,r,r'}}{\Gamma'_{x,r,r'}} = (-1)^{\inner{r}{x} \oplus \inner{r'}{x}}\left(|\alpha_{x,r,r'}|^2 - |\beta_{x,r,r'}|^2 - |\theta_{x,r,r'}|^2 + |\gamma_{x,r,r'}|^2\right).
\end{align*}
Now $\bob'$ and $\charlie'$ uncompute the unitary $U_B \otimes U_C$, and the state becomes \begin{align*}
&\left(U_B \otimes U_C\right)^{\dagger} \frac{1}{{|\cR|}} \sum_{r \in \cR} \ket{r}_\reg{B}\ket{r'}_\reg{C} \ket{\Gamma'_{x,r,r'}} =  \frac{1}{{|\cR|}} \sum_{r \in \cR} \ket{r}_\reg{B}\ket{r'}_\reg{C} \left(U^{r}_B \otimes U^{r'}_C\right)^{\dagger} \ket{\Gamma'_{x,r,r'}} \\
&= \frac{1}{{|\cR|}} \sum_{r,r' \in \cR} \ket{r}_\reg{B}\ket{r'}_\reg{C} \left((-1)^{\inner{r}{x} \oplus \inner{r'}{x}}\left(|\alpha_{x,r,r'}|^2 - |\beta_{x,r,r'}|^2 - |\theta_{x,r,r'}|^2 + |\gamma_{x,r,r'}|^2\right) \ket{\varphi_{x}}_\reg{BC} + \ket{\mathsf{err}_{x,r,r'}} \right), 
\end{align*}
where $\ket{\mathsf{err}_{x,r,r'}}$ is a subnormalized state orthogonal to $\ket{\varphi_{x}}_\reg{BC}$. \\

\par Next, $\bob'$ and $\charlie'$ each apply a Quantum Fourier Transform (QFT) on their random coins, resulting in the state \begin{align*}
    \frac{1}{{|\cR|^{2}}} &\sum_{r,r' \in \cR} \bigg((-1)^{\inner{r}{x} \oplus \inner{r'}{x}}\big(|\alpha_{x,r,r'}|^2 - |\beta_{x,r,r'}|^2 \nonumber \\
    &- |\theta_{x,r,r'}|^2 + |\gamma_{x,r,r'}|^2\big) \ket{\varphi_{x}}_\reg{BC} + \ket{\mathsf{err}_{x,r,r'}} \bigg) \left( \sum_{y,z \in \cR} (-1)^{\inner{r}{y} \oplus \inner{r'}{z}}\ket{y}_\reg{B}\ket{z}_\reg{C} \right).
\end{align*}
Note that the coefficient of $\ket{x}_\reg{B}\ket{x}_\reg{C}\ket{\varphi_{x}}$ equals \begin{align*}
    \frac{1}{|\cR|^{2}} \sum_{r,r' \in \cR} |\alpha_{x,r,r'}|^2 - |\beta_{x,r,r'}|^2 - |\theta_{x,r,r'}|^2 + |\gamma_{x,r,r'}|^2,
\end{align*}
so the probability that $\bob'$ and $\charlie'$ both output $x$ is lower bounded by \begin{align*}
    \pr{y = z = x} &\ge \E_{x}  \abs{ \frac{1}{|\cR|^{2}} \sum_{r,r' \in \cR} |\alpha_{x,r,r'}|^2 - |\beta_{x,r,r'}|^2 - |\theta_{x,r,r'}|^2 + |\gamma_{x,r,r'}|^2 }^2 \\
    &= \E_{x} \abs{ \E_{r,r'} |\alpha_{x,r,r'}|^2 - |\beta_{x,r,r'}|^2 - |\theta_{x,r,r'}|^2 + |\gamma_{x,r,r'}|^2 }^2 \\
    &\ge \abs{ \E_{x,r,r'} |\alpha_{x,r,r'}|^2 - |\beta_{x,r,r'}|^2 - |\theta_{x,r,r'}|^2 + |\gamma_{x,r,r'}|^2 }^2  \\
    &\ge 4\varepsilon^2,
\end{align*}
where we applied Jensen's inequality (due to the convexity of the square and absolute value functions) as well as \cref{eq:coeff_bound}.

\end{proof}

\end{document}